\documentclass[journal]{IEEEtran}
\usepackage{amsmath}
\usepackage{amssymb}
\usepackage{amsfonts}
\usepackage{graphicx}
\usepackage{epsfig}
\usepackage{subfigure}
\usepackage{psfrag}
\usepackage{color}
\usepackage{overpic}
\begin{document}

\title{Cognitive Beamforming Made Practical: Effective Interference Channel and Learning-Throughput Tradeoff}

\author{Rui Zhang,~\IEEEmembership{Member,~IEEE}, Feifei
Gao,~\IEEEmembership{Member,~IEEE}, and Ying-Chang
Liang,~\IEEEmembership{Senior Member,~IEEE}
\thanks{Manuscript received September 12, 2008; revised February 10, 2009 and May 5, 2009; accepted July 9, 2009.
This paper is presented in part at IEEE International Workshop on
Signal Processing Advances for Wireless Communications (SPAWC),
Perugia, Italy, June 21-24, 2009.}
\thanks{R. Zhang and Y.-C. Liang are with the Institute for Infocomm Research, A*STAR, Singapore (Email: \{rzhang,
ycliang\}@i2r.a-star.edu.sg).}
\thanks{F. Gao was with the Institute for Infocomm Research, A*STAR, Singapore.
He is now with School of Engineering and Science, Jacobs University,
Bremen, Germany (Email: feifeigao@ieee.org).}}

\maketitle

\begin{abstract}
This paper studies the transmit strategy for a secondary link or the
so-called cognitive radio (CR) link under opportunistic spectrum
sharing with an existing primary radio (PR) link. It is assumed that
the CR transmitter is equipped with multi-antennas, whereby transmit
precoding and power control can be jointly deployed to balance
between avoiding interference at the PR terminals and optimizing
performance of the CR link. This operation is named as
\emph{cognitive beamforming} (CB). Unlike prior study on CB that
assumes perfect knowledge of the channels over which the CR
transmitter interferes with the PR terminals, this paper proposes a
\emph{practical} CB scheme utilizing a new idea of \emph{effective
interference channel} (EIC), which can be efficiently estimated at
the CR transmitter from its observed PR signals. Somehow
surprisingly, this paper shows that the learning-based CB scheme
with the EIC improves the CR channel capacity against the
conventional scheme even with the exact CR-to-PR channel knowledge,
when the PR link is equipped with multi-antennas but only
communicates over a subspace of the total available spatial
dimensions. Moreover, this paper presents algorithms for the CR to
estimate the EIC over a finite learning time. Due to channel
estimation errors, the proposed CB scheme causes leakage
interference at the PR terminals, which leads to an interesting
\emph{learning-throughput tradeoff} phenomenon for the CR, pertinent
to its time allocation between channel learning and data
transmission. This paper derives the optimal channel learning time
to maximize the effective throughput of the CR link, subject to the
CR transmit power constraint and the interference power constraints
for the PR terminals.
\end{abstract}


\IEEEpeerreviewmaketitle

\newtheorem{condition}{\underline{Condition}}[section]
\newtheorem{theorem}{\underline{Theorem}}[section]
\newtheorem{lemma}{\underline{Lemma}}[section]
\newtheorem{proposition}{\underline{Proposition}}[section]
\newtheorem{example}{\underline{Example}}[section]
\newtheorem{remark}{\underline{Remark}}[section]
\newcommand{\mv}[1]{\mbox{\boldmath{$ #1 $}}}

\section{Introduction}

\PARstart{C}ognitive radio (CR), since the name was coined by Mitola
in his seminal work \cite{Mitola00}, has drawn intensive attentions
from both academic and industrial communities. Generally speaking,
there are three basic operation models for CRs, namely, {\it
Interweave}, {\it Overlay}, and {\it Underlay} (see, e.g.,
\cite{Goldsmith08} and references therein). The interweave method is
also known as {\it opportunistic spectrum access} (OSA), originally
outlined in \cite{Mitola00} and later introduced by DARPA, where the
CR is allowed to transmit over the spectrum allocated to an existing
primary radio (PR) system only when all PR transmissions are
detected to be off. In contrast to interweave, the overlay and
underlay methods allow the CR to transmit concurrently with PRs at
the same frequency. The overlay method utilizes an interesting
``cognitive relay'' idea \cite{Tarokh06}, \cite{Viswanath06}. For
this method, the CR transmitter is assumed to know perfectly all the
channels in the coexisting PR and CR links, as well as the PR
messages to be sent. Thereby, the CR transmitter is able to forward
PR messages to the PR receivers so as to compensate for the
interference due to its own messages sent concurrently to the CR
receiver. In comparison with overlay, the underlay method requires
only the channel gain knowledge from the CR transmitter to the PR
receivers, whereby the CR is permitted to transmit regardless of the
on/off status of PR transmissions provided that its resulted signal
power levels at all PR receivers are kept below some predefined
threshold, also known as the {\it interference-temperature}
constraint \cite{Haykin05}, \cite{Gastpar07}. From implementation
viewpoints, interweave and underlay methods could be more favorable
than overlay for practical CR systems.

In a wireless environment, channels are usually subject to
space-time-frequency variation (fading) due to multipath
propagation, mobility, and location-dependent shadowing. As such,
{\it dynamic resource allocation} (DRA) becomes crucial to CRs for
optimally deploying their transmit strategies, where the transmit
power, bit-rate, bandwidth, and antenna beam are dynamically
allocated based upon the channel state information (CSI) of the PR
and CR systems (see, e.g., \cite{Ghasemi07}-\cite{Zhang08a}). In
this paper, we are particularly interested in the case where the CR
terminal is equipped with multi-antennas so that it can deploy joint
transmit precoding and power control, namely {\it cognitive
beamforming} (CB), to effectively balance between avoiding
interference at the PR terminals and optimizing performance of the
CR link. In \cite{Zhang08a}, various CB schemes have been proposed
considering the CR transmit power constraint and a set of
interference power constraints at the PR terminals, under the
assumption that the CR transmitter knows perfectly all the channels
over which it interferes with PR terminals. In this work, however,
we propose a {\it practical} CB scheme, which does not require any
prior knowledge of the CR-to-PR channels. Instead, by exploiting the
time-division-duplex (TDD) operation mode of the PR link and the
channel reciprocities between the CR and PR terminals, the proposed
CB scheme utilizes a new idea so-called {\it effective interference
channel} (EIC), which can be efficiently estimated at the CR
terminal via periodically observing the PR transmissions. Thereby,
the proposed learning-based CB scheme eliminates the overhead for PR
terminals to estimate the CR-to-PR channels and then feed them back
to the CR, and thus makes the CB implementable in practical systems.

Furthermore, the proposed learning-based CB scheme with the EIC
creates a new operation model for CRs, where the CR is able to
transmit with PRs at the same time and frequency over the detected
available spatial dimensions, thus named as {\it opportunistic
spatial sharing} (OSS). On the one hand, OSS, like the underlay
method, utilizes the spectrum  more efficiently than the interweave
method by allowing the CR to transmit concurrently with PRs. On the
other hand, OSS can further improve the CR transmit spectral
efficiency over the underlay method by exploiting additional side
information on PR transmissions, which is extractable from the
observed EIC (more details will be given later in this paper).
Therefore, OSS is a more superior operation model for CRs than both
underlay and interweave methods in terms of the spectrum utilization
efficiency.

The main results of this paper constitute two parts, which are
summarized as follows:
\begin{itemize}

\item First, we consider the ideal case where the CR's
estimate on the EIC is {\it perfect} or noiseless. For this case, we
derive the conditions under which the EIC is sufficient for the
proposed CB scheme to cause no adverse effects on the concurrent PR
transmissions. In addition, we show that when the PR link is
equipped with multi-antennas but only communicates over a subspace
of the total available spatial dimensions, the learning-based CB
scheme with the EIC leads to a capacity gain over the conventional
zero-forcing (ZF) scheme \cite{Zhang08a} even with the exact
CR-to-PR channel knowledge, via exploiting side information on PR
transmit dimensions extracted from the EIC.

\item Second, we consider the practical case with {\it imperfect}
estimation of EIC due to finite learning time. We propose a {\it
two-phase} protocol for CRs to implement learning-based CB. The
first phase is for the CR to observe the PR signals and estimate the
EIC, while the second phase is for the CR to transmit data with CB
designed via the estimated EIC. We present two algorithms for CRs to
estimate the EIC, under different assumptions on the availability of
the noise power knowledge at the CR terminal. Furthermore, due to
imperfect channel estimation, the proposed CB scheme results in
leakage interference at the PR terminals, which leads to an
interesting {\it learning-throughput tradeoff}, i.e., different
choices of time allocation between CR's channel learning and data
transmission correspond to different tradeoffs between PR
transmission protection and CR throughput maximization. We formulate
the problem to determine the optimal time allocation for estimating
the EIC to maximize the effective throughput of the CR link, subject
to the CR transmit power constraint and the interference power
constraints at the PR terminals; and derive the solution via
applying convex optimization techniques.
\end{itemize}

The rest of this paper is organized as follows. Section
\ref{sec:system model} presents the CR system model. Section
\ref{sec:effective channel} introduces the idea of EIC. Section
\ref{sec:beamforming} studies the CB design based on the EIC under
perfect channel learning. Section \ref{sec:tradeoff} considers the
case with imperfect channel learning, presents algorithms for
estimating the EIC, and studies the learning-throughput tradeoff for
the CR link. Section \ref{sec:simulation results} presents numerical
results to corroborate the proposed studies. Finally, Section
\ref{sec:conclusion} concludes the paper.

{\it Notation}: Scalar is denoted by lower-case letter, e.g., $x$,
and bold-face lower-case letter is used for vector, e.g., $\mv{x}$,
and bold-face upper-case letter is for matrix, e.g., $\mv{X}$. For a
matrix $\mv{S}$, $\mathtt{Tr}(\mv{S})$, $\mathtt{Rank}(\mv{S})$,
$|\mv{S}|$, $\mv{S}^{-1}$, $\mv{S}^{\dag}$, $\mv{S}^{T}$, and
$\mv{S}^{H}$ denote its trace, rank, determinant, inverse, pseudo
inverse, transpose, and conjugate transpose, respectively.
$\mathtt{Diag}(x_1,\ldots,x_M)$ denotes a $M\times M$ diagonal
matrix with diagonal elements given by $x_1,\ldots,x_M$. For a
matrix $\mv{M}$, $\lambda_{\max}({\mv M})$ and $\lambda_{\min}({\mv
M})$ denote the maximum and minimum eigenvalues of ${\mv M}$,
respectively. $\mv{I}$ and $\mv{0}$ denote the identity matrix and
the all-zero matrix, respectively, with proper dimensions. For a
positive semi-definite matrix $\mv{S}$, denoted by
$\mv{S}\succcurlyeq \mv{0}$, $\mv{S}^{1/2}$ denotes a square-root
matrix of $\mv{S}$, i.e., $\mv{S}^{1/2}(\mv{S}^{1/2})^H=\mv{S}$,
which is assumed to be obtained from the eigenvalue decomposition
(EVD) of $\mv{S}$: If the EVD of $\mv{S}$ is expressed as
$\mv{S}=\mv{U}\mv{\Sigma}\mv{U}^H$, then
$\mv{S}^{1/2}=\mv{U}\mv{\Sigma}^{1/2}$. $\|\mv{x}\|$ denotes the
Euclidean norm of a complex vector $\mv{x}$. $\mathbb{C}^{x \times
y}$ denotes the space of $x\times y$ matrices with complex entries.
The distribution of a circular symmetric complex Gaussian (CSCG)
vector with mean vector $\mv{x}$ and covariance matrix $\mv{\Sigma}$
is denoted by $\mathcal{CN}(\mv{x},\mv{\Sigma})$, and $\sim$ stands
for ``distributed as''. $\mathbb{E}[\cdot]$ denotes the statistical
expectation. $\mathtt{Prob}\{\cdot\}$ denotes the probability.
$\max(x,y)$ and $\min(x,y)$ denote the maximum and the minimum of
two real numbers, $x$ and $y$, respectively. For a real number $a$,
$(a)^+=\max(0,a)$.

\section{System Model}\label{sec:system model}

\begin{figure}
\psfrag{a}{CR-Tx} \psfrag{b}{CR-Rx} \psfrag{c}{PR$_1$}
\psfrag{d}{PR$_2$} \psfrag{e}{$\mv{H}$} \psfrag{f}{$\mv{G}_1^H$}
\psfrag{g}{$\mv{G}_1$} \psfrag{h}{$\mv{G}_2$}
\psfrag{i}{$\mv{G}_2^H$} \psfrag{j}{$\mv{F}$} \psfrag{k}{$\mv{F}^H$}
\begin{center}
\scalebox{0.7}{\includegraphics*[50pt,530pt][387pt,740pt]{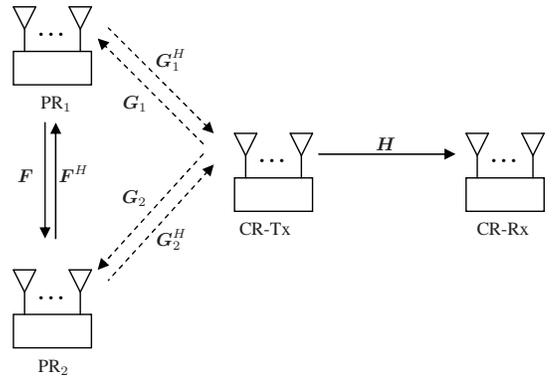}}
\end{center}
\caption{Spectrum sharing between a CR link and a PR link.}
\label{fig:system model}
\end{figure}

For the purpose of exposition, in this paper we consider a
simplified CR system as shown in Fig. \ref{fig:system model}, where
a single CR link consisting of one CR transmitter (CR-Tx) and one CR
receiver (CR-Rx) coexists with a single PR link consisting of two
terminals denoted by PR$_1$ and PR$_2$, respectively. The number of
antennas equipped at CR-Tx, CR-Rx, PR$_1$, and PR$_2$ are denoted by
$M_t$, $M_r$, $M_1$, and $M_2$, respectively. It is assumed that
$M_t>1$, while $M_r$, $M_1$, and $M_2$ can be any positive integers.
For the PR link, it is assumed that PR$_1$ and PR$_2$ operate in a
TDD mode over a narrow-band flat-fading channel. Furthermore,
reciprocity is assumed for the channels between PR$_1$ and PR$_2$,
i.e., if the channel from PR$_1$ to PR$_2$ is denoted by
$\mv{F}\in\mathbb{C}^{M_2\times M_1}$, then the channel from PR$_2$
to PR$_1$ becomes $\mv{F}^{H}$.\footnote{The results of this paper
hold similarly for the case where ${\mv F}^T$ instead of ${\mv F}^H$
is used to represent the reverse channel of $\mv{F}$.} Without loss
of generality (W.l.o.g.), the transmit beamforming matrix for
PR$_j$, $j=1,2$, is denoted by $\mv{A}_j\in\mathbb{C}^{M_j\times
d_j}$, with $d_j$ denoting the corresponding number of transmit data
streams, $1\leq d_j\leq M_j$. The transmit covariance matrix for
PR$_j$ is then defined as $\mv{S}_{j}\triangleq \mv{A}_j\mv{A}_j^H$.
We assume that $\mv{A}_j$ is a full-rank matrix and thus
$\mathtt{Rank}(\mv{S}_j)=d_j$. Furthermore, define
$\mv{B}_1\in\mathbb{C}^{d_2\times M_1}$ as the receive beamforming
matrix for PR$_1$ and $\mv{B}_2\in\mathbb{C}^{d_1\times M_2}$ for
PR$_2$. Both $\mv{B}_j$'s are assumed to be full-rank. It is also
assumed that PR$_1$ and PR$_2$ are both oblivious to the CR, and
treat the interference from the CR as additional noise.

The CR is assumed to transmit over the same frequency band of the
PR, and thus it needs to protect any active PR transmissions by
limiting its resulted interference power levels at both PR$_j$'s to
be below some prescribed threshold (to be specified later in Section
\ref{sec:tradeoff}). Let $\mv{H}\in\mathbb{C}^{M_r\times M_t}$
denote the CR channel, and $\mv{G}_j\in\mathbb{C}^{M_j\times M_t}$
denote the interference channel from CR-Tx to PR$_j$, $j=1,2$. Let
the transmit beamforming matrix of CR-Tx be denoted by a full-rank
matrix $\mv{A}_{\rm CR}\in\mathbb{C}^{M_t\times d_{\rm CR}}$, where
$d_{\rm CR}\leq M_t$ and $d_{\rm CR}=\mathtt{Rank}(\mv{S}_{\rm
CR})$, with $\mv{S}_{\rm CR}$ denoting the transmit covariance
matrix of CR-Tx, i.e., $\mv{S}_{\rm CR}\triangleq \mv{A}_{\rm
CR}\mv{A}_{\rm CR}^H$. Note that in Fig. \ref{fig:system model}, the
channels from PR$_j$'s to CR-Rx are not shown, while similarly as
for the PR terminals, it is assumed that any interference from
PR$_j$'s over these channels is treated as additional noise at
CR-Rx.

In \cite{Zhang08a}, various CB designs in terms of the CR transmit
covariance matrix, $\mv{S}_{\rm CR}$, have been studied for a
similar system setup like that in Fig. \ref{fig:system model}, where
the CR channel transmit rate is maximized under the CR transmit
power constraint and a set of interference power constraints for the
PR terminals. The CB designs in \cite{Zhang08a} assume that in Fig.
\ref{fig:system model}, CR-Tx has perfect knowledge of $\mv{H}$,
$\mv{G}_1$, and $\mv{G}_2$. In this work, however, we remove the
assumption on any prior knowledge of $\mv{G}_1$ and $\mv{G}_2$ at
CR-Tx for the design of CB, since in practice the CR and PR systems
usually belong to different operators, and it is thus difficult to
require PRs to estimate CR-to-PR channels and then feed them back to
the CR. As such, the most feasible method for the CR to learn some
knowledge of CR-to-PR channels is to observe the PR signals
propagating through PR-to-CR channels and then apply the channel
reciprocities between CR-Tx and PR$_j$'s. Thus, in this paper we
propose a {\it learning-based} CR transmit strategy, where the CR
first observes the received PR signals to extract CR-to-PR channel
knowledge, and then designs the CB based upon the obtained channel
knowledge. However, there are several issues related to this
approach, which are pointed out as follows:
\begin{itemize}

\item What CR-Tx can possibly estimate are indeed the ``effective''
channels, $\mv{G}_j^H\mv{A}_j$, from PR$_j$, $j=1,2$, instead of the
actual channels, $\mv{G}_j$'s, if the PR transmit beamforming
matrices, $\mv{A}_j$'s, are not known at CR-Tx.

\item The proposed CB schemes in \cite{Zhang08a} require that
the channels associated with $\mv{G}_1$ and $\mv{G}_2$ be separately
estimated at CR-Tx. As such, CR-Tx needs to synchronize with PR TDD
transmissions, which requires knowledge of the exact time instants
over each transmit direction between PR$_1$ and PR$_2$.

\item If CR-Tx designs $\mv{S}_{\rm CR}$ based on the effective
channels, $\mv{G}_j^H\mv{A}_j$'s, it is unclear whether the effect
of its resulted interference at PR$_j$'s can be properly controlled
since the signals from CR-Tx interferer with PR$_j$ via the
equivalent channel, $\mv{B}_j\mv{G}_j$, which can be different from
$\mv{A}_j^H\mv{G}_j$ if the PR receive beamforming matrix $\mv{B}_j$
differs from $\mv{A}_j$.
\end{itemize}

Therefore, to make the learning-based CB feasible for practical
systems, the above issues need to be carefully addressed, without
critical assumptions or prior knowledge on PR signal processing
procedures. In this paper, we propose an effective solution to
resolve the aforementioned issues, while utilizes a new idea, namely
{\it effective interference channel} (EIC), as will be presented
next.

\section{Effective Interference Channel}\label{sec:effective
channel}

For the learning-based CB scheme, suppose that prior to data
transmission, CR-Tx first listens to the frequency band of interest
for PR transmissions over $N$ symbol periods. The received baseband
signals can then be represented as
\begin{equation}\label{eq:received signal CR-Tx}
\mv{y}(n)=\mv{G}_j^H\mv{A}_j\mv{t}_j(n)+\mv{z}(n), \ n=1,\ldots,N
\end{equation}
where $j=1$ if $n \in\mathcal{N}_1$, and $j=2$ if $n
\in\mathcal{N}_2$, with
$\mathcal{N}_1,\mathcal{N}_2\subseteq\{1,\ldots,N\}$ denoting the
time instants when PR$_1$ transmits to PR$_2$ and PR$_2$ transmits
to PR$_1$, respectively, and
$\mathcal{N}_1\cap\mathcal{N}_2=\varnothing$ due to the assumed TDD
mode; $\mv{t}_j(n)$'s are the encoded signals (prior to power
control and precoding) for the corresponding PR$_j$, and for
convenience it is assumed that $\mv{t}_j(n)$'s are independent over
$n$'s and $\mathbb{E}[\mv{t}_j(n)(\mv{t}_j(n))^H]=\mv{I}_{d_j\times
d_j}$, $j=1,2$; $\mv{z}(n)$'s are the additive noises assumed to be
independent CSCG random vectors with zero-mean and covariance matrix
denoted by $\rho_0\mv{I}_{M_t\times M_t}$. Denote the cardinality of
the set ${\cal N}_j$ as $|{\cal N}_j|$. It is reasonable to assume
that PR$_j$ will transmit, with a constant probability $\alpha_j<1$,
during a certain time period. Mathematically, we may use
$\mathbb{E}\left[\frac{|{\cal N}_j|}{N}\Big|N\right]=\alpha_j$ or
$\mathbb{E}\left[\frac{|{\cal N}_j|}{N}\right]=\alpha_j$. Note that
$\alpha_1+\alpha_2\leq 1$, where a strict inequality occurs when
there are guard (silent) intervals between alternate PR TDD
transmissions. Also note that if $\alpha_1=\alpha_2=0$, there will
be no active PR transmissions in the observed frequency band.

Define ${\mv s}_j(n)$ as $q_j(n){\mv t}_j(n)$, where $q_j(n)=1$, if
$n\in{\cal N}_j$ and $q_j(n)=0$ otherwise. Obviously, $q_j(n)$'s are
random variables with $\mathbb{E}[q_j(n)]=\alpha_j$. Meanwhile,
$q_1(n)$ and $q_2(n)$ are related  by $q_1(n)q_2(n)=0, \forall n$.
Thus, we have $\mathbb{E}[{\mv s}_j(n)({\mv s}_j(n))^H]=\alpha_j{\mv
I}$, $j=1,2$, but $\mathbb{E}[{\mv s}_1(n)({\mv s}_2(n))^H]={\mv
0}$. The signal model in (\ref{eq:received signal CR-Tx}) can then
be equivalently rewritten as
\begin{equation}\label{eq:new_received signal CR-Tx}
\mv{y}(n)={\mv {\cal A}}{\mv s}(n)+{\mv z}(n),\ n=1,\ldots,N
\end{equation}
where $\mv{{\cal A}}=[{\mv G}_1^H{\mv A}_1, {\mv G}_2^H{\mv A}_2]$
and $\mv{s}(n)=[({\mv s}_1(n))^T, ({\mv s}_2(n))^T]^T$. The
covariance matrix of the received signals at CR-Tx is then defined
as
\begin{equation}\label{eq:covariance matrix}
\mv{Q}_y=\mathbb{E}[\mv{y}(n)(\mv{y}(n))^H]=\mv{Q}_s+\rho_0{\mv I}
\end{equation}
where
\begin{equation}\label{eq:PR signal covariance matrix}
\mv{Q}_s\triangleq\alpha_1\mv{G}_1^H\mv{S}_1\mv{G}_1+\alpha_2\mv{G}_2^H\mv{S}_2\mv{G}_2
\end{equation}
denotes the covariance matrix due to only the signals from PR$_j$'s.

Practically, only the sample covariance matrix can be obtained at
CR-Tx, which is expressed as
\begin{equation}\label{eq:sample covariance matrix}
\hat{\mv{Q}}_y=\frac{1}{N}\sum_{n=1}^N\mv{y}(n)(\mv{y}(n))^H.
\end{equation}
From law of large number (LLN), it is easy to verify that
$\hat{\mv{Q}}_y\rightarrow \mv{Q}_s+\rho_0\mv{I}$ with probability
one as $N\rightarrow \infty$, while for finite values of $N$,
$\mv{Q}_s$  can only be estimated from
$\hat{\mv{Q}}_y$.\footnote{Algorithms for such an estimation are
discussed in Section \ref{subsec:channel estimation}.} Denote
$\hat{\mv{Q}}_s$ as the estimate of $\mv{Q}_s$ from
$\hat{\mv{Q}}_y$. Note that $\hat{\mv{Q}}_s$ should be a covariance
matrix and hence $\hat{\mv{Q}}_s\succcurlyeq \mv{0}$ and
$\hat{\mv{Q}}_s^H=\hat{\mv{Q}}_s$. Next, we denote the aggregate
``effective'' channel from both PR$_j$'s to CR-Tx as
\begin{equation}\label{eq:effective channel}
\mv{G}_{\rm eff}^H=\hat{\mv{Q}}_s^{1/2}
\end{equation}
while under the assumption of channel reciprocity, we denote the
{\it effective interference channel} (EIC) from CR-Tx to both
PR$_j$'s as $\mv{G}_{\rm eff}$.

In the rest of this paper, CB schemes based on the EIC instead of
the actual CR-to-PR channels will be studied. Note that with the
EIC, the first two items of implementation issues raised in Section
\ref{sec:system model} are resolved. The first issue is resolved
since the EIC is defined over the effective channels from PR$_j$'s
to CR-Tx instead of the actual channels, while the second issue is
resolved since the EIC does not attempt to separate the two channels
from PR$_j$'s to CR-Tx, and thus synchronization for CR-Tx with each
transmit direction between PR$_1$ and PR$_2$ is no longer required.
However, we still need to address the third issue on analyzing the
effect of the CR's interference on the PR transmissions with the
EIC-based CB design. We will first address this issue for the ideal
case where the estimation of $\mv{G}_{\rm eff}$ is perfect or
noiseless in Section \ref{sec:beamforming}, in order to gain some
insights into this problem. Then, we will study this problem for the
more practical case where $\mv{G}_{\rm eff}$ is imperfectly
estimated due to finite values of $N$ in Section \ref{sec:tradeoff}.

\section{Perfect Channel Learning}\label{sec:beamforming}

In this section, we design the CR transmit covariance matrix,
$\mv{S}_{\rm CR}$, in terms of the equivalent transmit beamforming
matrix, $\mv{A}_{\rm CR}$, which contains information of both
transmit precoding and power allocation at CR-Tx, under perfect
learning of the EIC, i.e., the noise effect on estimating $\mv{Q}_s$
from $\hat{\mv{Q}}_y$ is completely removed. For this case,
$\hat{\mv{Q}}_s=\mv{Q}_s$ in (\ref{eq:effective channel}), and from
(\ref{eq:PR signal covariance matrix}) it follows that the EIC can
now be expressed as
\begin{equation}\label{eq:EIC ideal case}
\mv{G}_{\rm eff}
=\left(\left(\alpha_1\mv{G}_1^H\mv{S}_1\mv{G}_1+\alpha_2\mv{G}_2^H\mv{S}_2\mv{G}_2\right)^{1/2}\right)^H.
\end{equation}
From (\ref{eq:EIC ideal case}), due to independence of the channels
$\mv{G}_1$ and  $\mv{G}_2$, it follows that $d_{\rm
eff}=\mathtt{Rank}(\mv{G}_{\rm eff})=\min(d_1+d_2, M_t)$. Thus, if
the number of antennas at CR-Tx, $M_t$, is strictly greater than the
total number of transmit data streams between PR$_1$ and PR$_2$,
$d_1+d_2$, then the EIC-based CB design will have at most
$M_t-(d_1+d_2)$ number of spatial dimensions or degrees of freedom
(DoF) \cite{Paulraj} for transmission, where all these dimensions
lie in the null space of $\mv{G}_{\rm eff}$. Based on this
observation, we obtain the following proposition:
\begin{proposition}\label{proposition:1}
Under perfect learning of the EIC, if the conditions
$\mv{A}_j^H\mv{G}_j \sqsupseteq \mv{B}_j\mv{G}_j, j=1,2$
hold,\footnote{$\mv{X}\sqsupseteq \mv{Y}$ means that for two given
matrices with the same collum size, $\mv{X}$ and $\mv{Y}$, if
$\mv{Xe}=\mv{0}$ for any arbitrary vector $\mv{e}$, then
$\mv{Ye}=\mv{0}$ must hold.} then the EIC-based CB design satisfying
the constraint $\mv{G}_{\rm eff}\mv{A}_{\rm CR}=\mv{0}$ will have no
adverse effects on PR transmissions, i.e.,
$\mv{B}_j\mv{G}_j\mv{A}_{\rm CR}=\mv{0}, j=1,2$.
\end{proposition}

The conditions in the above proposition can also be expressed as
${\rm Span}({\mv A}_j^H{\mv G}_j)\supseteq{\rm Span}({\mv B}_j{\mv
G}_j)$, $j=1,2$, where ${\rm Span}({\mv X})$ denotes the subspace
spanned by the row vectors in ${\mv X}$. Intuitively speaking, these
conditions hold when the transmit signal space of PR$_j$ after
propagating through the PR-to-CR channel $\mv{G}_j^H$, i.e., ${\mv
G}_j^H{\mv A}_j$, if being reversed (conjugate transposed), will
subsume the equivalent channel from CR-Tx to PR$_j$, ${\mv B}_j{\mv
G}_j$, as a subspace, for both $j=1,2$. Note that
$\mv{A}_j^H\mv{G}_j$ and $\mv{B}_j\mv{G}_j$ may not have the same
column size, and $\mv{A}_j^H$ and $\mv{B}_j$ may differ from each
other for any $j=1,2$. As such, the validity of these conditions
needs to be examined for practical systems. Thus, before we proceed
to the proof of Proposition \ref{proposition:1}, we present two
typical examples of multi-antenna transmission schemes for the PR
link as follows, for both of which the conditions
$\mv{A}_j^H\mv{G}_j \sqsupseteq \mv{B}_j\mv{G}_j, j=1,2$ are usually
satisfied.\footnote{Note that when the conditions in Proposition
\ref{proposition:1} are not satisfied, the proposed CB scheme will
cause certain performance loss to PR transmissions even under
perfect channel learning.}

\begin{example}
{\it Spatial Multiplexing}: When the PR CSI is unknown at
transmitter but known at receiver, the spatial multiplexing mode is
usually adopted to assign equal powers and rates to transmit
antennas (e.g., the V-BLAST scheme \cite{Paulraj}). For this case,
the transmit covariance matrix at PR$_j$, $j=1,2$, reduces to
$\mv{S}_j=\frac{P_j}{M_j}\mv{I}_{M_j\times M_j}$, with $P_j$
denoting the transmit power of PR$_j$. Thus, $d_j=M_j$, and
$\mv{A}_j$'s are both scaled identity matrices. It then follows that
$\mv{A}_j^H\mv{G}_j \sqsupseteq \mv{B}_j\mv{G}_j, j=1,2,$ holds
regardless of $\mv{B}_j$.
\end{example}

\begin{example}\label{example:eigenmode transmission}
{\it Eigenmode Transmission}: When the PR CSI is known at both
transmitter and receiver, which is usually a valid assumption for
the TDD mode, the eigenmode transmission mode is usually adopted to
decompose the multi-antenna PR channel into parallel scalar channels
\cite{Paulraj}. For this case, $\mv{S}_1$ and $\mv{S}_2$ are
designed based on the singular-value decomposition (SVD) of $\mv{F}$
and $\mv{F}^H$, respectively. Let the SVD of $\mv{F}$ be
$\mv{U}_F\mv{\Sigma}_F\mv{V}_F^H$. It then follows that
$\mv{A}_1=\mv{V}_{F(1)}\mv{\Lambda}_{1}^{1/2}$,
$\mv{B}_1=\mv{V}_{F(2)}^H$,
$\mv{A}_2=\mv{U}_{F(2)}\mv{\Lambda}^{1/2}_{2}$, and
$\mv{B}_2=\mv{U}_{F(1)}^H$, where
$\mv{\Lambda}_{j}=\mathtt{Diag}(\lambda_{j,1},\ldots,\lambda_{j,d_j})$
is a positive diagonal matrix with $\lambda_{j,i}, i=1,\dots,d_j,$
denoting the power allocation over the $i$th transmit data stream,
and $\mv{V}_{F(j)}$ ($\mv{U}_{F(j)}$) consists of the first $d_j$
columns in $\mv{V}_F$ ($\mv{U}_F$). Note that in this case
$d_j\leq\min(M_1,M_2)$. If it is true that $d_1=d_2$, i.e., both
transmit directions between PR$_1$ and PR$_2$ have the same number
of data streams, then it follows that $\mv{A}_j^H\mv{G}_j
\sqsupseteq \mv{B}_j\mv{G}_j$ holds for both $j=1,2$. Note that a
valid special case here is the ``beamforming mode'' \cite{Paulraj}
with $d_1=d_2=1$.
\end{example}

Next, we present the proof of Proposition \ref{proposition:1} as
follows:

\begin{proof}
First, with perfectly known $\mv{G}_{\rm eff}$, $\mv{G}_{\rm
eff}\sqsupseteq \mv{A}_j^H\mv{G}_j$ is true for $j=1,2$. This can be
shown as follows given any arbitrary vector $\mv{e}$: $\mv{G}_{\rm
eff}\mv{e}=\mv{0}\overset{(a)}{\Rightarrow}(\hat{\mv{Q}}_s^{1/2})^H\mv{e}=0\overset{(b)}{\Rightarrow}
(\mv{Q}_s^{1/2})^H\mv{e}=0 \Rightarrow \mv{e}^H\mv{Q}_s\mv{e}=0
\overset{(c)}{\Rightarrow} \|\mv{A}_j^H\mv{G}_j\mv{e}\|^2=0, j=1,2
\Rightarrow \mv{A}_j^H\mv{G}_j\mv{e}=\mv{0}, j=1,2$, where $(a)$ is
from (\ref{eq:effective channel}), $(b)$ is due to
$\hat{\mv{Q}}_s=\mv{Q}_s$, and $(c)$ is from (\ref{eq:PR signal
covariance matrix}). Since for arbitrary matrices $\mv{X},\mv{Y},$
and $\mv{Z}$, $\mv{X}\sqsupseteq\mv{Y}$ and
$\mv{Y}\sqsupseteq\mv{Z}$ imply that $\mv{X}\sqsupseteq\mv{Z}$, from
$\mv{G}_{\rm eff}\sqsupseteq \mv{A}_j^H\mv{G}_j$ (shown above) and
$\mv{A}_j^H\mv{G}_j \sqsupseteq \mv{B}_j\mv{G}_j$ (given in
Proposition \ref{proposition:1}) it follows that $\mv{G}_{\rm
eff}\sqsupseteq \mv{B}_j\mv{G}_j$, $j=1,2$. Therefore, if the
constraint $\mv{G}_{\rm eff}\mv{A}_{\rm CR}=\mv{0}$ is satisfied, it
follows that $\mv{B}_j\mv{G}_j\mv{A}_{\rm CR}=\mv{0}, j=1,2$, i.e.,
the interference from CR-Tx at PR$_j$ lies in the null space of the
corresponding receiver beamforming matrix $\mv{B}_j$, and thus has
no adverse effects.
\end{proof}

From Proposition \ref{proposition:1}, it is known that if the given
conditions are satisfied, it is sufficient for us to design
$\mv{A}_{\rm CR}$ subject to the constraint $\mv{G}_{\rm
eff}\mv{A}_{\rm CR}=\mv{0}$, in order to remove the effects of the
CR signals on PR transmissions. Let the EVD of $\mv{Q}_s$ be
represented as $\mv{Q}_s=\mv{V}\mv{\Sigma}\mv{V}^{H}$, where
$\mv{V}\in\mathbb{C}^{M_t\times d_{\rm eff}}$ and $\mv{\Sigma}$ is a
positive $d_{\rm eff}\times d_{\rm eff}$ diagonal matrix. Note that
$\mathtt{Rank}(\mv{Q}_s)=d_{\rm eff}$. From (\ref{eq:effective
channel}), $\mv{G}_{\rm eff}^H$ can then be represented as
$\mv{G}_{\rm eff}^H=\mv{V}\mv{\Sigma}^{1/2}$. Define the projection
matrix related to $\mv{V}$ as
$\mv{P}_V\triangleq\mv{I}-\mv{V}\mv{V}^H={\mv U}{\mv U}^H$, where
$\mv{U}\in\mathbb{C}^{M_t\times (M_t-d_{\rm eff})}$ satisfies ${\mv
V}^H\mv{U}=\mv{0}$. We are now ready to present the general form of
$\mv{A}_{\rm CR}$ satisfying the constraint $\mv{G}_{\rm
eff}\mv{A}_{\rm CR}=\mv{0}$ as \cite{Spencer04}
\begin{equation}\label{eq:optimal CR beamforming}
\mv{A}_{\rm CR}={\mv U}\mv{C}_{\rm CR}^{1/2}
\end{equation}
where $\mv{C}_{\rm CR}^{1/2}\in\mathbb{C}^{(M_t-d_{\rm eff})\times
d_{\rm CR}}$ with $d_{\rm CR}$ denoting the number of transmit data
streams of the CR, and $\mv{C}_{\rm CR}\in\mathbb{C}^{(M_t-d_{\rm
eff})\times (M_t-d_{\rm eff})}$ satisfies that $\mv{C}_{\rm
CR}\succcurlyeq \mv{0}$ and $\mathtt{Tr}(\mv{C}_{\rm
CR})=\mathtt{Tr}(\mv{S}_{\rm CR})\leq P_{\rm CR}$, with $P_{\rm CR}$
denoting the transmit power constraint of CR-Tx. From
(\ref{eq:optimal CR beamforming}), it follows that designing the
transmit beamforming matrix $\mv{A}_{\rm CR}$ for the CR channel
becomes equivalent to designing the transmit covariance matrix
$\mv{C}_{\rm CR}$ for an auxiliary multi-antenna channel,
$\mv{H}\mv{U}$, subject to transmit power constraint,
$\mathtt{Tr}(\mv{C}_{\rm CR})\leq P_{\rm CR}$. This observation
simplifies the design for the remaining part in $\mv{A}_{\rm CR}$,
i.e., $\mv{C}_{\rm CR}$, since existing solutions (see, e.g.,
\cite{Paulraj} and references therein) are available for this
well-studied precoder design problem.

At last, we demonstrate an interesting property for the proposed CB
scheme given in (\ref{eq:optimal CR beamforming}), when the
conditions given in Proposition \ref{proposition:1} are satisfied,
and furthermore, PR$_1$ and/or PR$_2$ have multi-antennas but
transmit only over a subspace of the available spatial dimensions,
i.e., $d_j<\min(M_1,M_2), j=1,2$. For this case, we will show that
the proposed scheme in (\ref{eq:optimal CR beamforming}) with the
EIC $\mv{G}_{\rm eff}$ can be superior over the conventional
``projected-channel SVD (P-SVD)'' scheme proposed in \cite{Zhang08a}
with the actual CR-to-PR channels $\mv{G}_1$ and $\mv{G}_2$, in
terms of the achievable DoF for CR transmission. At a first glance,
this result is counter-intuitive since $\mv{G}_{\rm eff}$ contains
only partial information on $\mv{G}_j$'s. The key observation here
is that $\mv{G}_{\rm eff}$ contains information on
$\mv{A}_j^H\mv{G}_j$'s, which also exhibit side information on
$\mv{B}_j\mv{G}_j$'s via the conditions, $\mv{A}_j^H\mv{G}_j
\sqsupseteq \mv{B}_j\mv{G}_j, j=1,2$, given in Proposition
\ref{proposition:1}, while $\mv{B}_j\mv{G}_j$'s are assumed to be
unknown for the P-SVD scheme. More specifically, for the proposed
scheme, the DoF is given as $d_{\rm CR}$, which can be shown to be
upper-bounded by $\min(M_t-d_{\rm
eff},M_r)=\min((M_t-d_1-d_2)^+,M_r)$. In comparison with the
proposed scheme, the P-SVD scheme with perfect knowledge of
$\mv{G}_1$ and $\mv{G}_2$ removes the interference (thus having no
effects on PR transmissions like the proposed scheme) at both PR$_1$
and PR$_2$ via transmitting only over the subspace of $\mv{H}$ that
is orthogonal to both $\mv{G}_1$ and $\mv{G}_2$, thus resulting in
the DoF to be at most $\min((M_t-M_1-M_2)^+,M_r)$. Therefore, the
proposed scheme can have a strictly positive DoF even when
$M_1+M_2\geq M_t$, provided that $d_1+d_2<M_t$, i.e., the total
number of antennas of PR$_j$'s is no smaller than $M_t$, while the
total number of data streams over both transmit directions between
PR$_1$ and PR$_2$ is smaller than $M_t$, while the P-SVD scheme has
a zero DoF in this case since $M_t\leq M_1+M_2$. In most practical
cases, we have $d_j\leq\min(M_1,M_2), j=1,2$. It thus follows that
$(d_1+d_2) \leq (M_1+M_2)$ and thus  the DoF gain of the proposed
scheme against the P-SVD scheme,
$\min((M_t-d_1-d_2)^+,M_r)-\min((M_t-M_1-M_2)^+,M_r)$, is always
non-negative, while the maximum DoF gain occurs when $d_1=d_2=0$,
i.e., when the PR link switches off transmissions. Note that this
DoF gain is achieved by the CR via exploiting side information on PR
transmit (on/off) status or signal dimensions extracted from the
observed EIC. This also justifies our previous claim that the OSS
operation mode with learning-based CB is potentially more spectral
efficient than the conventional interweave and underlay methods.

\begin{figure}[t]
        \centering
        \includegraphics*[width=9.5cm]{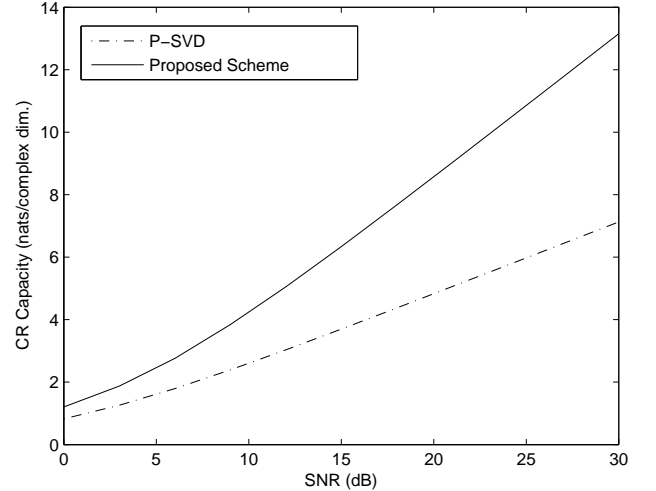}
        \caption{CR capacity comparison for the proposed CB scheme and the P-SVD scheme in \cite{Zhang08a}.}
        \label{fig:transmit_gain}
\end{figure}

\begin{example}
The capacity gain of the proposed scheme in (\ref{eq:optimal CR
beamforming}) over the P-SVD scheme in \cite{Zhang08a}, as above
discussed, is shown in Fig. \ref{fig:transmit_gain} for a PR link
with $M_1=M_2=2$, $d_1=d_2=1$ (i.e., beamforming mode corresponding
to the largest channel singular value in Example
\ref{example:eigenmode transmission} is used), and a CR link with
$M_t=5$ and $M_r=3$. All the channels involved are assumed to have
the standard Rayleigh-fading distribution, i.e., each element of the
channel matrix is independent CSCG random variable $\sim
\mathcal{CN}(0,1)$. For simplicity, it is assumed that the
interference due to PR transmissions at CR-Rx is included in the
additive noise, which is assumed to be
$\mathcal{CN}(\mv{0},\rho_1\mv{I})$. The signal-to-noise ratio (SNR)
in this case is thus defined as $P_{\rm CR}/\rho_1$. The DoF can be
visually seen in the figure to be proportional to the asymptotic
ratio between the capacity value over the log-SNR value as SNR goes
to infinity \cite{Paulraj}. It is observed that the DoF for the
proposed scheme is approximately three times of that for the P-SVD
scheme in this case, since
$\min((M_t-d_1-d_2)^+,M_r)/\min((M_t-M_1-M_2)^+,M_r)=3/1=3$.
\end{example}

\section{Imperfect Channel Learning}\label{sec:tradeoff}

\begin{figure}
\psfrag{a}{\hspace{-0.1in}Channel Learning} \psfrag{b}{Data
Transmission}\psfrag{d}{\hspace{-0.2in}$T-\tau$} \psfrag{c}{$\tau$}
\begin{center}
\scalebox{0.8}{\includegraphics*[72pt,685pt][461pt,736pt]{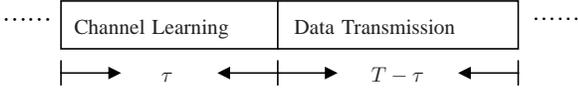}}
\end{center}
\caption{Two-phase protocol for CR block-based transmission.}
\label{fig:protocol}
\end{figure}

In the previous section, CB designs have been studied under the
assumption that the EIC $\mv{G}_{\rm eff}$ is perfectly estimated at
CR-Tx. In this section, we will study the effect of imperfect
estimation of $\mv{G}_{\rm eff}$ due to finite sample size $N$ on
the performance of the proposed CB scheme. First, consider the
following {\it two-phase} protocol for CRs to support the
learning-based CB scheme as shown in Fig. \ref{fig:protocol}, where
each block transmission of the CR with duration $T$ is divided into
two consecutive sub-blocks. During the first sub-block of duration
$\tau$, the CR observes the PR transmissions and estimates
$\mv{G}_{\rm eff}$; during the second sub-block of duration
$T-\tau$, the CR transmits data with the CB design based on the
estimated $\mv{G}_{\rm eff}$ (following the same procedure as in
Section \ref{sec:beamforming}, but with the estimated $\mv{G}_{\rm
eff}$ instead of its true value). Note that $T$ needs to be chosen
such that, on the one hand, to be sufficiently small compared with
the channel coherence time in order to maintain constant channels
during each transmission block, and on the other hand, to be as
large as possible compared with the inverse of the channel bandwidth
in order to make $T$ constitute a large number of symbols to reduce
the overhead of channel learning. In this paper, we assume that $T$
is preselected as a fixed value. For a given $T$, intuitively, a
larger value of $\tau$ is desirable from the perspective of
estimating $\mv{G}_{\rm eff}$, while a smaller $\tau$ is favorable
in terms of the effective CR link throughput that is proportional to
$(T-\tau)/T$. Therefore, there exists a general {\it
learning-throughput tradeoff} for the proposed CB
scheme,\footnote{Note that the learning-throughput tradeoff includes
the sensing-throughput tradeoff studied in \cite{Liang08} as a
special case since CR channel sensing to detect PR's on/off status
\cite{Liang08} can be considered as a ``hard'' version of CR channel
learning proposed in this paper.} where different choices for the
value $\tau$ lead to different tradeoffs between PR transmission
protection and CR throughput maximization.

The rest of this section is organized as follows. We first present
two practical algorithms for CRs to estimate $\mv{G}_{\rm eff}$ with
a finite $N$ in Section \ref{subsec:channel estimation}. Then, in
Section \ref{subsec:interference power} we analyze the so-called
``effective leakage interference'' at PR terminals for the proposed
CB scheme due to estimation errors in $\mv{G}_{\rm eff}$. At last,
in Section \ref{subsec:optimization problem} we characterize the
learning-throughput tradeoff for CRs by finding the optimal value of
$\tau$ to maximize the CR link effective throughput, under the given
$T$, $P_{\rm CR}$, and maximum tolerable leakage interference power
level at PR$_j$'s.

\subsection{Estimation of ${\mv G}_{\rm eff}$} \label{subsec:channel estimation}
From (\ref{eq:effective channel}), it is known that $\mv{G}_{\rm
eff}$ depends solely on ${\hat {\mv Q}}_s$, which is the estimate of
the received PR signal covariance matrix ${\mv Q}_s$ defined in
(\ref{eq:PR signal covariance matrix}). Thus, in this subsection, we
present two algorithms to obtain ${\hat {\mv Q}}_s$ from the
received sample covariance matrix ${\hat{\mv Q}_y}$ given in
(\ref{eq:sample covariance matrix}). Denote the EVD of ${\hat{\mv
Q}_y}$ as
\begin{equation}
{\hat{\mv Q}_y}={\hat{\mv T}_y}{\hat {\mv \Lambda}_y}{\hat {\mv
T}}_y^H
\end{equation}
where ${\hat {\mv \Lambda}_y}={\rm Diag}({\hat \lambda_1}, {\hat
\lambda_2},\ldots,{\hat \lambda_{M_t}})$ is a positive diagonal
matrix whose diagonal elements are the eigenvalues of ${\hat{\mv
Q}}_y$. W.l.o.g., we assume that ${\hat \lambda}_i$'s,
$i=1,\ldots,M_t$, are arranged in a decreasing order. We then obtain
${\hat{\mv Q}}_s$ from ${\hat{\mv Q}}_y$ based on the standard
maximum likelihood (ML) criterion, for the following two cases:

\subsubsection{Known noise power $\rho_0$}

In the case where the noise power, $\rho_0$, is assumed to be known
at CR-Tx prior to channel learning, it follows from \cite{Lim} that
the ML estimate of ${\mv Q}_s$ is obtained as
\begin{equation}\label{eq:channel estimate 1}
{\hat {\mv Q}}_s={\hat {\mv T}}_y{\rm Diag}\left(({\hat
\lambda}_1-\rho_0)^+,\ldots,({\hat
\lambda}_{M_t}-\rho_0)^+\right){\hat {\mv T}}_y^H.
\end{equation}
The rank of ${\hat {\mv Q}}_s$, or the estimate of $d_{\rm eff}$,
denoted as ${\hat d}_{\rm eff}$, can be found as the largest integer
such that ${\hat \lambda}_{{\hat d}_{\rm eff}}>\rho_0$. Therefore,
the first ${\hat d}_{\rm eff}$ columns of ${\hat {\mv T}}_y$ give
the estimate of ${\mv V}$, denoted by ${\hat {\mv V}}$, and the last
$M_t-{\hat d}_{\rm eff}$ columns of ${\hat {\mv T}}_y$ are deemed as
the estimate of ${\mv U}$, denoted by ${\hat {\mv U}}$. Note that
${\hat {\mv U}}$ will replace the true value of ${\mv U}$ in
(\ref{eq:optimal CR beamforming}) for the proposed CB design in the
case with imperfect channel learning.

\subsubsection{Unknown noise power $\rho_0$}
In this case, $\rho_0$ is unknown to CR-Tx and has to be estimated
along with ${\hat{\mv Q}}_s$. The ML estimate of $\rho_0$ can first
be obtained as \cite{Kailath}
\begin{equation}
{\hat \rho}_0=\frac{1}{M_t-{\hat d}_{\rm eff}}\sum_{i={\hat d}_{\rm
eff}+1}^{M_t}{\hat \lambda}_i
\end{equation}
where ${\hat d}_{\rm eff}$ is the ML estimate of $d_{\rm eff}$.
Specifically, ${\hat d}_{\rm eff}$ can be obtained as \cite{Kailath}
\begin{align}
{\hat d}_{\rm eff}=&\arg\max_k\
(M_t-k)N\log\left(\frac{\prod_{i=k+1}^{M_t}{\hat
\lambda}_i^{1/(M_t-k)}}{\frac{1}{M_t-k}\sum_{i=k+1}^{M_t}{\hat\lambda}_i}\right)
\nonumber \\ =&\arg\max_k\ (M_t-k)N\log\left(\frac{{\rm GM}(k)}{{\rm
AM}(k)}\right) \label{eq:rank2}
\end{align}
where ${\rm GM}(k)$ and ${\rm AM}(k)$ denote the geometric mean and
the arithmetic mean of the last $M_t-k$ eigenvalues of ${\hat{\mv
Q}}_y$, respectively. To make this estimation unbiased, we
conventionally adopt the so-called minimum description length (MDL)
estimator expressed as \cite{Kailath}
\begin{equation}
{\hat d}_{\rm eff}=\arg\min_k\ (M_t-k)N\log\left(\frac{{\rm
AM}(k)}{{\rm GM}(k)}\right)+\frac{1}{2}k(2M_t-k)\log N
\end{equation}
where the second term on the right-hand side (RHS) is a bias
correction term. The ML estimates of ${\mv V}$ and ${\mv U}$,
denoted by ${\hat {\mv V}}$ and ${\hat {\mv U}}$, are then obtained
from the first ${\hat d}_{\rm eff}$ and the last $M_t-{\hat d}_{\rm
eff}$ columns of ${\hat {\mv T}}_y$, respectively.

After knowing ${\hat \rho}_0$, ${\hat d}_{\rm eff}$, ${\hat {\mv
V}}$, and ${\hat {\mv U}}$, the ML estimate of ${\mv Q}_s$ is
obtained as
\begin{equation}\label{eq:channel estimate 2}
{\hat {\mv Q}}_s={\hat {\mv V}}{\rm Diag}\left({\hat
\lambda}_1-{\hat \rho}_0,\ldots,{\hat \lambda}_{{\hat d}_{\rm
eff}}-{\hat\rho}_0\right){\hat {\mv V}}^H.
\end{equation}

From (\ref{eq:channel estimate 1}) and (\ref{eq:channel estimate
2}), it is observed that the above two estimators have a similar
structure while they differ in the noise power term adopted and the
way to estimate the rank of ${\mv Q}_s$, $d_{\rm eff}$.

\subsection{Effective Leakage Interference} \label{subsec:interference power}

Due to imperfect channel estimation, the CB design in
(\ref{eq:optimal CR beamforming}) with $\mv{U}$ replaced by ${\hat
{\mv U}}$ cannot perfectly remove the effective interference at
PR$_j$'s. In this subsection, the effect of channel estimation
errors on the resultant leakage interference power levels at
PR$_j$'s will be analytically quantified so as to assist the later
studies. Define the rank over-estimation probability as
$p_o(k)=\mathtt{Prob}({\hat d}_{\rm eff}-d_{\rm eff}=k|{\hat d}_{\rm
eff})$, $k=1,\ldots,{\hat d}_{\rm eff}$, and the rank
under-estimation probability as $p_u(k)=\mathtt{Prob}(d_{\rm
eff}-{\hat d}_{\rm eff}=k|{\hat d}_{\rm eff})$, $k=1,\ldots,
M_t-{\hat d}_{\rm eff}$, both conditioned on the observation ${\hat
d}_{\rm eff}$. If the over-estimation of $d_{\rm eff}$ is
encountered,  the upper bound on the number of data streams from
CR-Tx, $d_{\rm CR}$, may be affected. However, as long as
$(M_t-{\hat d}_{\rm eff})\geq M_r$, $d_{\rm CR}$ is more tightly
bounded by $M_r$ and the over-estimation of $d_{\rm eff}$ does not
cause any problem. On the other hand, the under-estimation of
$d_{\rm eff}$ will bring a severe issue, since some columns in
${\hat {\mv U}}$ may actually come from the PR signal subspace
spanned by ${\mv V}$. In this case, the amount of interference at
PRs will be tremendously increased, which is similar to the scenario
in the conventional interleave-based CR system when a misdetection
of active PR transmissions occurs to the CR. In practice, a
threshold $\xi$ should be properly set, and the last $M_t-({\hat
d}_{\rm eff}+k_0)$ columns in ${\hat{\mv T}}_y$ are chosen as
${\hat{\mv U}}$ only if $p_o(k_0)\geq \xi$.

Detailed study on $p_o(k)$, $p_u(k)$, and $\xi$ is deemed as a
separate topic of this paper and will not be further addressed here.
In this paper, for simplicity we will assume that the rank of ${\mv
Q}_s$ or $d_{\rm eff}$ is correctly estimated. We will then focus on
studying the effect of finite $N$ on the distortion of the estimated
eigenspace $\hat {\mv U}$. From (\ref{eq:optimal CR beamforming}),
the transmit signal at CR-Tx in the case of imperfect channel
learning is expressed as
\begin{equation}\label{eq:s CR}
{\mv s}_{\rm CR}(n)={\mv A}_{\rm CR}\mv{t}_{\rm CR}(n)=\hat{\mv U}
{\mv C}_{\rm CR}^{1/2}\mv{t}_{\rm CR}(n), \quad n>N
\end{equation}
where ${\mv s}_{\rm CR}(n)$ is the precoded version of the data
vector ${\mv t}_{\rm CR}(n)$. Note that $\mathbb{E}[{\mv t}_{\rm
CR}(n)({\mv t}_{\rm CR}(n))^H]={\mv I}$ and $\mathbb{E}[{\mv s}_{\rm
CR}(n)({\mv s}_{\rm CR}(n))^H]=\mv{S}_{\rm CR}$.  The average
leakage interference power at PR$_j, j=1,2$, due to the CR
transmission is then expressed as
\begin{eqnarray}\label{eq:Ij}
I_j={\mathbb E}[\|{\mv B}_j{\mv G}_j{\mv s_{CR}}(n)\|^2].
\end{eqnarray}
Next, $I_j$ is normalized by the corresponding processed (via
$\mv{B}_j$) noise power to unify the discussions for PR$_j$'s. For
convenience, it is assumed that the additive noise power at PR$_j$
is equal to $\rho_0$, the same as that at CR-Tx, and thus the
processed noise power becomes $\rho_0\mathtt{Tr}({\mv B}_j{\mv
B}_j^H)$. Define
\begin{align}\label{eq:Ij bar}
\bar{I}_{j}\triangleq\frac{I_j}{\rho_0\mathtt{Tr}({\mv B}_j{\mv
B}_j^H)}.
\end{align}
$\bar{I}_j$ is then named as the ``effective leakage interference
power'' at PR$_j$ since it measures the power of interference
normalized by that of noise after they are both processed by the
receive beamforming matrix, $\mv{B}_j$.
\begin{lemma}\label{lemma:upperbounds}
The upper bounds on $\bar{I}_{j}$, $j=1,2$, are given as
\begin{align}\label{eq:UB Ij bar}
\bar{I}_{j}\leq\frac{\mathtt{Tr}({\mv
C}_{CR})}{\alpha_jN}\frac{\lambda_{\max}({\mv G}_j{\mv
G}_j^H)}{\lambda_{\min}({\mv A}^H_j{\mv G}_j{\mv G}_j^H{\mv A}_j)}.
\end{align}
\end{lemma}
\begin{proof}
Please refer to Appendix \ref{appendix:proof upperbounds}.
\end{proof}
From Lemma \ref{lemma:upperbounds}, it follows that the upper bound
on $\bar{I}_{j}$ is proportional to the CR transmit power
$\mathtt{Tr}({\mv C}_{CR})$, but inversely proportional to
$\alpha_j$, $N$, and the PR$_j$'s average transmit power $P_j$ (via
${\mv A}_j$). Some nice properties of the derived effective leakage
interference powers associated with the proposed CB scheme are
listed as follows:
\begin{itemize}
\item $\bar{I}_{j}$ is upper-bounded by a finite value provided that $\alpha_j>0$.\footnote{Note that the derived upper
bound on $\bar{I}_{j}$ is practically meaningful when $\alpha_j$ is
a non-negligible positive number, since in the extreme case of
$\alpha_j=0$, PR$_j$ switches off its transmissions and as a result
discussion for the interference at PR$_j$ becomes irrelevant.} Note
that $\lambda_{\min}({\mv A}^H_j{\mv G}_j{\mv G}_j^H{\mv A}_j)>0$ if
$M_t>d_j$ and in this case ${\mv A}_j^H{\mv G}_j$ is a full-rank and
fat matrix.
\item $\bar{I}_{j}$ can be easily shown to be invariant to any scalar multiplication over
${\mv G}_j$. Thus, the CR protects PR$_j$'s equally regardless of
their location-dependent signal attenuation.
\item Since for fixed $N$ and $P_{\rm CR}$ the upper bound on $\bar{I}_{j}$ is inversely
proportional to $\alpha_j$ and $P_j$, PR$_j$ gets better protected
if it transmits more frequently and/or with more power. This
property is useful for the CR to design {\it fair} rules for
distributing its leakage interference among the coexisting PRs.
\end{itemize}

\begin{example}\label{example:I/N relation}
In Figs. \ref{fig:IT} (a) and \ref{fig:IT} (b), numerical results on
$\bar{I}_j$'s given in (\ref{eq:Ij bar}) as well as theoretical
results on the upper bounds on $\bar{I}_j$'s given in (\ref{eq:UB Ij
bar}) are compared for PR SNR being $15$dB and $0$dB, respectively.
Note that $P_1=P_2=P$ in this example and PR SNR is defined as
$P/\rho_0$, where $\rho_0$ is assumed to be known at CR-Tx. For the
PR, it is assumed that $M_1=M_2=1$, $\alpha_1=0.3$, and
$\alpha_2=0.6$, while for the CR, $M_t=4$, $P_{\rm CR}=100$, and
$\mv{C}_{\rm CR}$ is designed for the eigenmode transmission.
$2,000$ random channel realizations are used for averaging while the
standard Rayleigh fading channel distribution is adopted. To clearly
see the effect of $N$, we take the inverses of $\bar{I}_j$'s or
their upper bounds for the vertical axis of each figure. It is
observed that at high-SNR region, the theoretical and numerical
results match well, and the interference powers are inversely
linearly proportional to $N$. However, at low-SNR region, there
exists big mismatch between the two results. This is reasonable
since the first order approximation of (\ref{eq:detlaU}) in Appendix
\ref{appendix:proof upperbounds} is inaccurate at low-SNR region.
Nonetheless, the good news is that the inverse of interference power
is observed to be still linearly proportional to $N$ from the
numerical results.
\end{example}

\begin{figure}
  \centering
  \subfigure[SNR$=15$dB]{
    \label{subfig:IT1} 
    \includegraphics[width=80mm]{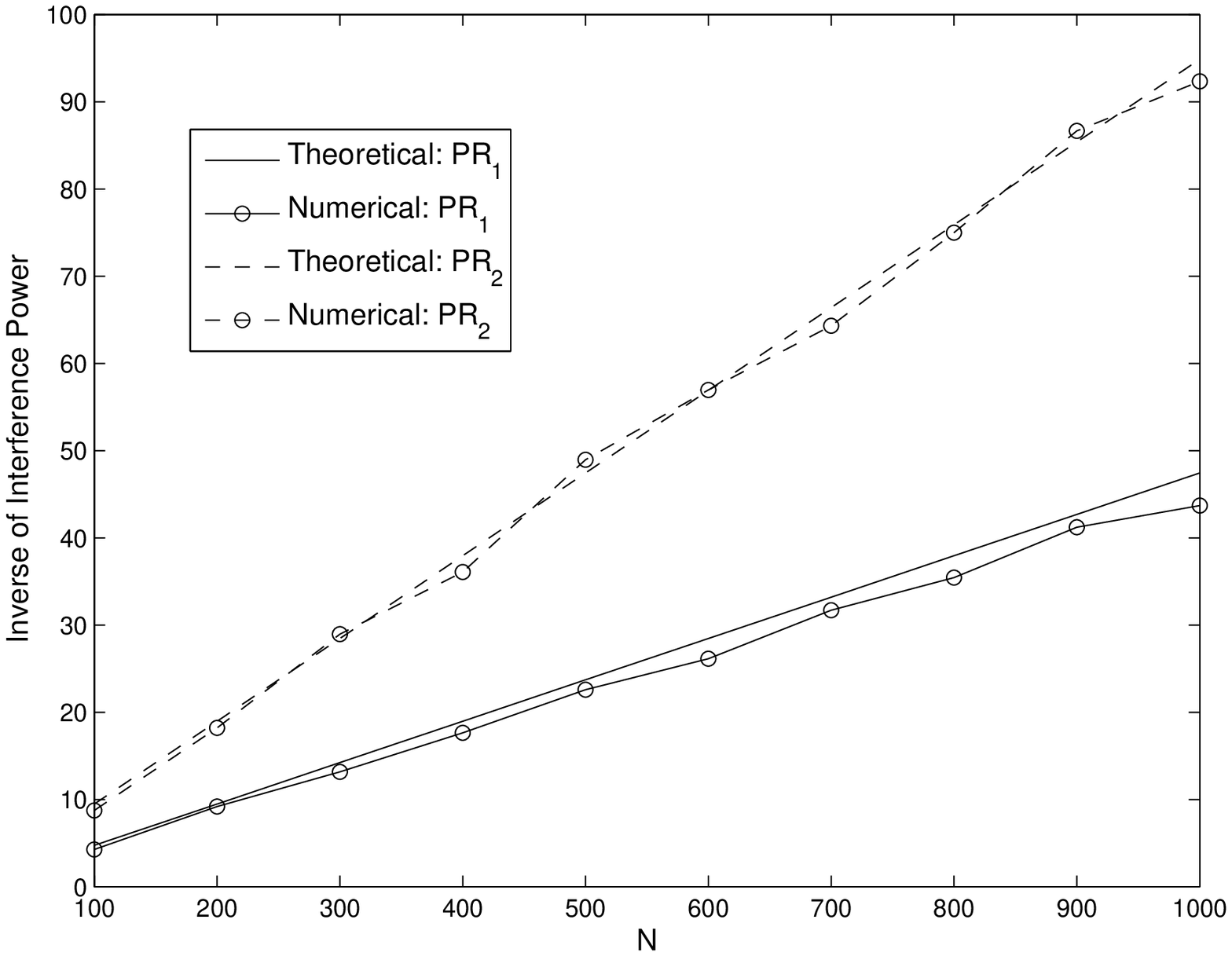}}
  \subfigure[SNR$=0$dB]{
    \label{subfig:IT2} 
    \includegraphics[width=80mm]{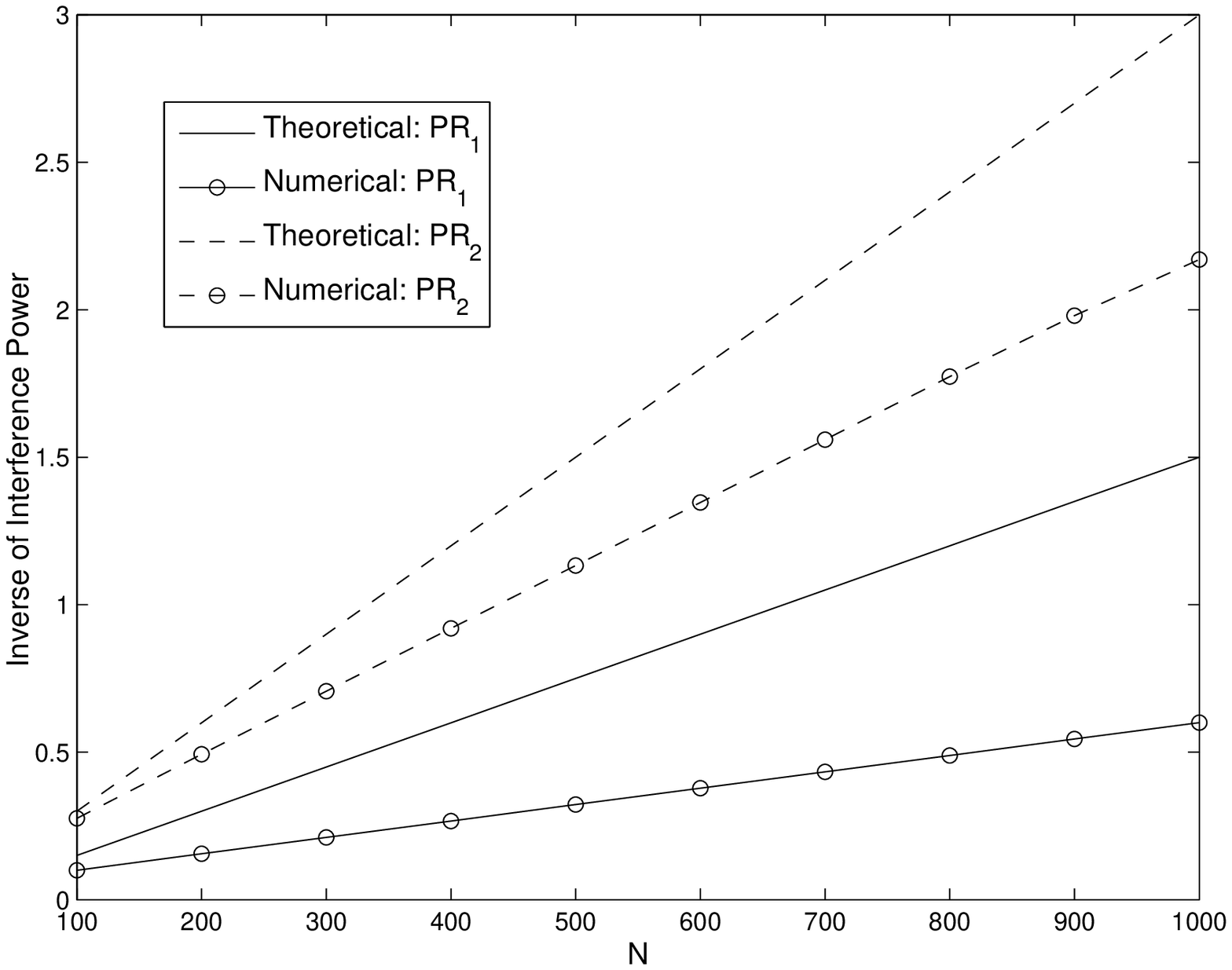}}
  \caption{Leakage interference power levels at PR$_1$ and PR$_2$ for different PR SNRs.}
  \label{fig:IT} 
\end{figure}

\subsection{Optimal Learning Time} \label{subsec:optimization problem}

At last, we study the learning-throughput tradeoff for CRs by
determining the optimal learning time $\tau$ for a given $T$ to
maximize the CR link throughput, subject to both transmit power
constraint of the CR and effective leakage interference power
constraints at the two PR terminals. It is assumed that the CR
channel $\mv{H}$ is known at both CR-Tx and CR-Rx. From
(\ref{eq:optimal CR beamforming}) with $\mv{U}$ replaced by
$\hat{\mv{U}}$, the maximum CR link effective throughput, assuming
the receiver noise $\sim\mathcal{CN}(\mv{0},\rho_1\mv{I})$, is
expressed as \cite{Coverbook}
\begin{equation}
\frac{T-\tau}{T}\log\left|{\mv I}+{\mv H}{\hat {\mv U}}{\mv
C}_{CR}{\hat {\mv U}}^H{\mv H}^H/\rho_1\right|
\end{equation}
where the term $(T-\tau)/T$ accounts for the portion of throughput
loss due to channel learning.

If the {\it peak} transmit power constraint for the CR is adopted,
we have $\mathtt{Tr}({\mv C}_{CR})\leq P_{CR}$, while if the {\it
average} transmit power constraint is adopted, we may allocate the
total power for each block to the second phase transmission,
resulting in $\mathtt{Tr}({\mv C}_{CR})\leq \frac{T}{T-\tau}P_{CR}$.
Let $\Gamma$ denote the prescribed constraint on the maximum
effective leakage interference powers, $\bar{I}_j$'s defined in
(\ref{eq:Ij bar}). Note that $N$ is related to $\tau$ by
$N=\tau/T_s$, where $T_s$ denotes the symbol period. From Lemma
\ref{lemma:upperbounds}, it follows that it is sufficient for ${\mv
C}_{CR}$ to satisfy the following inequality to ensure the given
interference power constraints:
\begin{equation}\label{eq:interference constraint}
\mathtt{Tr}({\mv C}_{CR})\leq\gamma_j\tau, \ j=1,2
\end{equation}
where
\begin{equation}\label{eq:gamma j}
\gamma_j=\frac{\zeta_j\alpha_j\Gamma}{T_s}\frac{\lambda_{\min}({\mv
A}^H_j{\mv G}_j{\mv G}_j^H{\mv A}_j)}{\lambda_{\max}({\mv G}_j{\mv
G}_j^H)}
\end{equation}
and $\zeta_j$, $\zeta_j\leq1$, is an additional margin that accounts
for any analytical errors (e.g., approximations made at low-SNR
region in Example \ref{example:I/N relation}). In practice, the
choice of $\gamma_j$'s in (\ref{eq:gamma j}) depends on the
calibration process at CR-Tx, based on prior knowledge of
$\zeta_j$'s, $\Gamma$, and $T_s$, as well as the observed average
signal power from PRs.

Let $\gamma=\min(\gamma_1, \gamma_2)$. Then, the interference power
constraints in (\ref{eq:interference constraint}) become equivalent
to $\mathtt{ Tr}({\mv C}_{CR})\leq\gamma\tau$. The problem for
maximizing the CR effective throughput is thus expressed as
\begin{align}
({\rm P1}):\quad \max_{\tau, {\mv C}_{CR}}\quad &\frac{T-\tau}{T}\log\left|{\mv I}+{\mv H}{\hat {\mv U}}{\mv C}_{CR}{\hat {\mv U}}^H{\mv H}^H/\rho_1\right| \nonumber \\
\mathtt{ s.t.}\quad &\mathtt{ Tr}({\mv C}_{CR})\leq J,\quad {\mv
C}_{CR}\succcurlyeq \mv{0},\quad 0\leq\tau< T \nonumber
\end{align}
where $J=\min(P_{CR},\ \gamma\tau)$ for the case of peak transmit
power constraint, while $J=\min\left(\frac{T}{T-\tau}P_{CR},\
\gamma\tau\right)$ for the case of average transmit power
constraint.

For problem (P1), it is noted that $\hat{\mv{U}}$ is related to
$\tau$, which makes the maximization over $\tau$ difficult. However,
it can be verified that the matrix norm of $\Delta {\mv
U}=\hat{\mv{U}}-\mv{U}$ decreases in the order of
$\mathcal{O}(1/\sqrt{\tau})$, as compared to the norm of $\mv{U}$.
Therefore, the overall term $\hat{\mv{U}}=\mv{U}+\Delta\mv{U}$ in
the objective function is dominated by $\mv{U}$, and changes slowly
with $\tau$ when $\tau$ is sufficiently large.  Thus, we assume that
the effect of $\tau$ on $\hat{\mv{U}}$ is ignored in the subsequent
analysis, and will verify this assumption via simulation results in
Section \ref{sec:simulation results}.

Let the EVD of ${\hat{\mv U}}^H{\mv H}^H{\mv H}{\hat{\mv U}}$ be
${\mv U}_h{\mv \Sigma}_h{\mv U}_h^H$, where ${\mv U}_h$ is a
$(M_t-d_{\rm eff})\times (M_t-d_{\rm eff})$ unitary matrix and ${\mv
\Sigma}_h=\mathtt{Diag}(\sigma_{h,1}^2,\ldots,\sigma_{h,M_t-d_{\rm
eff}}^2)$. W.l.o.g., we assume that $\sigma_{h,i}^2$'s are arranged
in a descending order. Note that if $(M_t-d_{\rm eff})>M_r$, then
$\sigma_{h,i}$'s, $i=M_r+1,\ldots,M_t-d_{\rm eff}$, all have zero
values. Define ${\mv X}$ as ${\mv U}_h^H{\mv C}_{\rm CR}{\mv U}_h$.
Problem (P1) is then converted to
\begin{align}
({\rm P2}):\quad\max_{\tau, {\mv X}}\quad
&\frac{T-\tau}{T}\log\left|{\mv I}+{\mv X}{\mv
\Sigma}_h/\rho_1\right| \nonumber \\
\mathtt{ s.t.}\quad & \mathtt{ Tr}({\mv X})\leq J, \quad {\mv
X}\succcurlyeq \mv{0},\quad 0\leq\tau< T \nonumber
\end{align}
where the optimal ${\mv C}_{\rm CR}$ can be later recovered as ${\mv
U}_h{\mv X}{\mv U}_h^H$. By the standard approach like in
\cite[Chapter 10.5]{Coverbook}, it can be shown that the optimal
${\mv X}$ is a diagonal matrix
$\mv{X}=\mathtt{Diag}(x_1,\ldots,x_{M_t-d_{\rm eff}})$ and $x_i$'s,
$i=1,\dots, M_t-d_{\rm eff}$, are obtained from
\begin{align}
({\rm P3}):\quad\max_{\tau, \{x_{i}\}}\quad
&\frac{T-\tau}{T}\sum_{i=1}^{M_t-d_{\rm eff}}
\log\left(1+\frac{\sigma^2_{h,i}x_i}{\rho_1}\right) \nonumber\\
\mathtt{ s.t.}\quad & \sum_{i=1}^{M_t-d_{\rm eff}}x_i\leq J,\quad
x_i\geq 0, \quad 0\leq\tau < T. \nonumber
\end{align}

Next, we will study (P3) for the cases of peak and average transmit
power constraints, respectively.

\subsubsection{Peak transmit power constraint}

In this case, if $P_{CR}>\gamma T$, then $J$ is always equal to
$\gamma\tau$. Therefore, we consider the more general case with
$P_{CR}\leq \gamma T$. The remaining discussion will then be divided
into the following two parts for $P_{\rm CR}/\gamma<\tau< T$ and
$0\leq\tau\leq P_{\rm CR}/\gamma$, respectively.

If $P_{\rm CR}/\gamma<\tau< T$, then $J=P_{\rm CR}$ and the
optimization in problem (P3) over $\tau$ and $x_i$'s can be
separated. The optimization over $x_i$'s directly follows the
conventional water-filling (WF) algorithm \cite{Coverbook}. For the
ease of later discussion, we define
\begin{align}\label{eq:optimize f z}
f(z)=\max_{\{x_{i}\}}\quad &\sum_{i=1}^{M_t-d_{\rm
eff}}\log\left(1+\frac{\sigma^2_{h,i}x_i}{\rho_1}\right) \nonumber \\
\mathtt{ s.t.}\quad & \sum_{i=1}^{M_t-d_{\rm eff}}x_i\leq z, \quad
x_i\geq 0.
\end{align}
The WF solution of the above optimization problem is then given as
$x_i=(\frac{1}{\mu}-\frac{\rho_1}{\sigma_{h,i}^2})^+$, where
$\frac{1}{\mu}$ is the water level that should satisfy
\begin{equation}\label{eq:water level}
\sum_{i=1}^{M_t-d_{\rm
eff}}\left(\frac{1}{\mu}-\frac{\rho_1}{\sigma^2_{h,i}}\right)^+=z.
\end{equation}

Denote
$q_k=\frac{k\rho_1}{\sigma^2_{h,k+1}}-\sum_{i=1}^k\frac{\rho_1}{\sigma^2_{h,i}}$,
for $k=0,\ldots, M_t-d_{\rm eff}$. Obviously, $q_0=0$, and
$q_{M_t-d_{\rm eff}}=+\infty$ since $\sigma^2_{h,M_t-d_{\rm eff}+1}$
is set to be zero. Then, we can express $f(z)$ as
\begin{equation}\label{eq:f z}
f(z)=\sum_{i=1}^k\log\left(\frac{\sigma^2_{h,i}}{k\rho_1}\left(z+\sum_{i=1}^k\frac{\rho_1}{\sigma^2_{h,i}}\right)\right),
\quad z\in[q_{k-1},q_{k}].
\end{equation}
Note that $k$ is the number of dimensions assigned with positive
$x_i$'s. The objective function of problem (P3) in this case can
then be explicitly written as
\begin{equation}\label{eq:g1}
g_1(\tau)\triangleq\frac{T-\tau}{T} f(P_{\rm CR}).
\end{equation}
Since $\frac{T-\tau}{T}$ is a decreasing function of $\tau$, the
optimal $\tau$ to maximize $g_1(\tau)$ over $P_{\rm
CR}/\gamma<\tau\leq T$ is simply $P_{CR}/\gamma$.

Next, consider $0\leq\tau\leq P_{\rm CR}/\gamma$. In this case,
$J=\gamma\tau$, and problem (P3) becomes
\begin{equation}\label{eq:g2}
\max_{0\leq\tau\leq P_{\rm CR}/\gamma}\quad
g_2(\tau)\triangleq\frac{T-\tau}{T}f(\gamma\tau).
\end{equation}
In order to study the function $g_2(\tau)$, some properties of the
function $f(z)$ are given below.
\begin{lemma}\label{lemma:concave}
$f(z)$ is a continuous, increasing, differentiable, and concave
function of $z$.
\end{lemma}
\begin{proof}
Please refer to Appendix \ref{appendix:proof concavity}.
\end{proof}
With Lemma \ref{lemma:concave}, it can be easily verified that
$g_2(\tau)$ is also a continuous, differentiable, and concave
function of $\tau$. Thus, the optimal value of $\tau$, denoted as
$\tau_2^*$, to maximize $g_2(\tau)$ can be easily obtained via,
e.g., the Newton method \cite{Boydbook}.

To summarize the above two cases, the optimal solution of $\tau$ for
problem (P3) in the case of peak transmit power constraint can be
obtained as
\begin{eqnarray}
\tau^*=\left\{\begin{array}{ll} \tau_2^*, & \tau_2^*<P_{\rm
CR}/\gamma \\  P_{CR}/\gamma, & {\rm otherwise}. \end{array}\right.
\end{eqnarray}
The above solution is illustrated in Fig. \ref{fig:optimal_tau}. The
optimal value of (P3) then becomes $g_2(\tau_2^*)$ if
$\tau_2^*<P_{\rm CR}/\gamma$, and $g_1(P_{CR}/\gamma)$ otherwise.

\begin{figure}
  \centering
  \subfigure[$\tau^*=\tau_2^*<P_{CR}/\gamma$]{
    \label{subfig:op1} 
    \begin{overpic}[width=80mm]{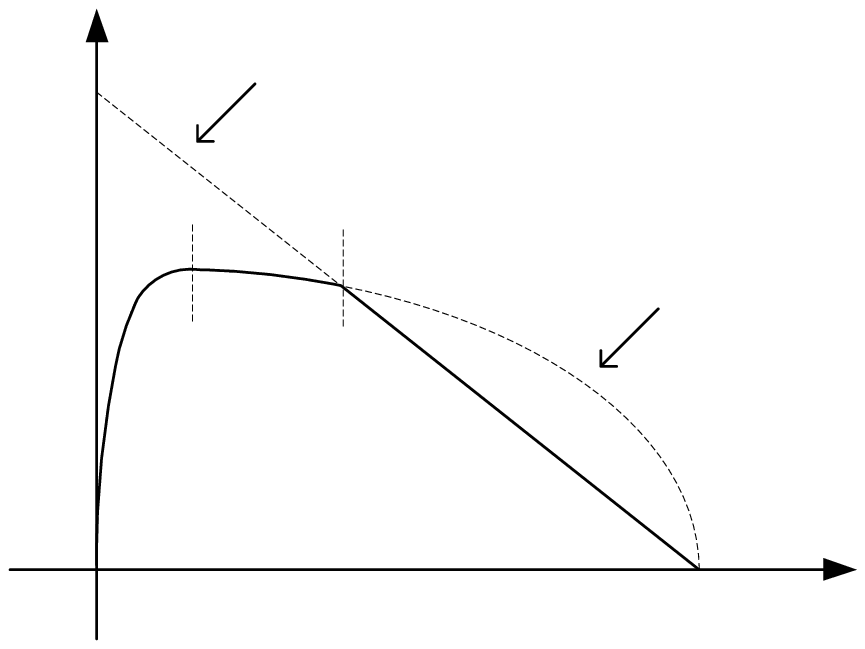}
    \put(31,67){$g_1(\tau)$}
    \put(76,41.5){$g_2(\tau)$}
    \put(39,33){$\frac{P_{CR}}{\gamma}$}
    \put(80,5){$T$}
    \put(98,6){$\tau$}
    \put(21,34){$\tau_2^*$}
  \end{overpic}}
  \hspace{10mm}
  \subfigure[$\tau^*=P_{CR}/\gamma<\tau_2^*$]{
    \label{subfig:op2} 
    \begin{overpic}[width=80mm]{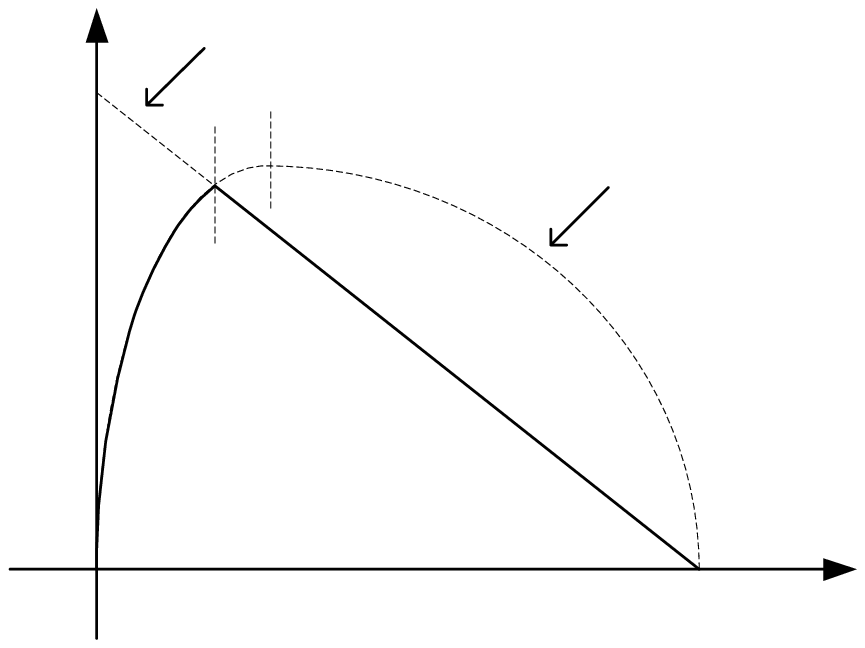}
    \put(25,70){$g_1(\tau)$}
    \put(71,54){$g_2(\tau)$}
    \put(22,43){$\frac{P_{CR}}{\gamma}$}
    \put(80,5){$T$}
    \put(98,6){$\tau$}
    \put(31,62){$\tau^*_2$}
  \end{overpic}}
  \caption{Illustration of the optimal learning time $\tau^*$ for the case of peak CR transmit-power constraint.}
  \label{fig:optimal_tau} 
\end{figure}

\subsubsection{Average transmit power constraint} In this case, $J$ in
problem (P3) takes the value of $T/(T-\tau)P_{CR}$ if
$T/(T-\tau)P_{CR}< \gamma\tau$, and $\gamma\tau$ otherwise. It can
be verified that $T/(T-\tau)P_{CR}< \gamma\tau$ for some $\tau$ in
$[0,T)$ only if $P_{\rm CR}/\gamma<T/4$. In other words, if $P_{\rm
CR}/\gamma\geq T/4$, $J$ always takes the value $\gamma\tau$
regardless of $\tau$. Thus, the objective function of (P3) is always
given as $g_2(\tau)$, and the optimal solution of $\tau$ is
$\tau_2^*$.

Therefore, we consider next the more general case of $P_{\rm
CR}/\gamma<T/4$. For this case, it can be shown that the equation
$T/(T-\tau)P_{CR}= \gamma\tau$ always has two positive roots of
$\tau$, denoted as $\tau_l$ and $\tau_u$, respectively, and $0\leq
\tau_l<\tau_u<T$. If $0\leq\tau\leq\tau_l$ or $\tau_u\leq\tau<T$,
$J$ takes the value of $\gamma\tau$, and then the maximum value of
(P3) is obtained by the $\tau$ that maximizes $g_2(\tau)$ over this
interval of $\tau$. Otherwise, the maximum value occurs when $\tau$
is given as
\begin{equation}
\arg\max_{\tau, \tau_l<\tau<\tau_u} g_3(\tau)\triangleq
\frac{T-\tau}{T}f\left(\frac{T}{T-\tau}P_{\rm CR}\right).
\end{equation}
It can be shown that $g_3(\tau)$ is a continuously decreasing
function of $\tau$, for $\tau\in[0,T)$. Thus, the optimal value of
$\tau$ to maximize $g_3(\tau)$ over this interval of $\tau$ is
simply $\tau_l$.

To summarize the above discussions, we obtain the optimal solution
of $\tau$ for problem (P3) in the case of average transmit power
constraint as
\begin{eqnarray}
\tau^*=\left\{\begin{array}{ll} \tau_2^*, & \tau_2^*<\tau_l
\\  \tau_l, & {\rm otherwise}. \end{array}\right.
\end{eqnarray}
The above solution is illustrated in Fig. \ref{fig:optimal_tau avg}.
The optimal value of (P3) then becomes $g_2(\tau_2^*)$ if
$\tau_2^*<\tau_l$, and $g_3(\tau_l)$ otherwise.

\begin{figure}
  \centering
  \subfigure[$\tau^*=\tau_2^*<\tau_l$]{
    \label{subfig:op1 new} 
    \begin{overpic}[width=80mm]{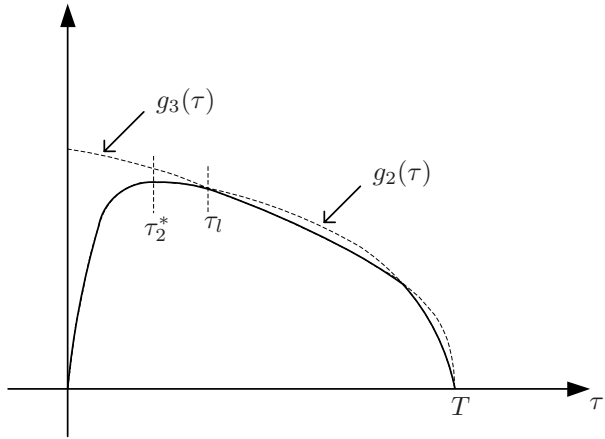}
    \put(26,56){$g_3(\tau)$}
    \put(62,44){$g_2(\tau)$}
    \put(34,36){$\tau_l$}
    \put(75,5){$T$}
    \put(98,6){$\tau$}
    \put(24,35){$\tau_2^*$}
  \end{overpic}}
  \hspace{10mm}
  \subfigure[$\tau^*=\tau_l<\tau_2^*$]{
    \label{subfig:op2 new} 
    \begin{overpic}[width=80mm]{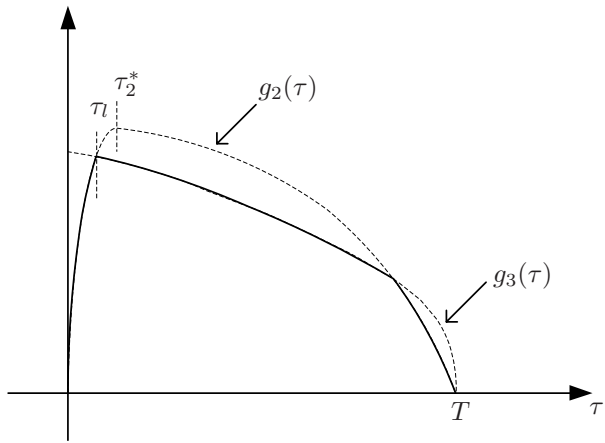}
    \put(82,28){$g_3(\tau)$}
    \put(43,59){$g_2(\tau)$}
    \put(15,56){$\tau_l$}
    \put(75,5){$T$}
    \put(98,6){$\tau$}
    \put(19,60){$\tau^*_2$}
  \end{overpic}}
  \caption{Illustration of the optimal learning time $\tau^*$ for the case of average CR transmit-power constraint.}
  \label{fig:optimal_tau avg} 
\end{figure}

\section{Numerical Results}\label{sec:simulation results}

In the section, we present additional simulation results to
demonstrate the performance of the proposed CB scheme under
imperfect channel learning. The system parameters are taken as
$M_t=6$, $M_r=3$, $M_1=4$, and $M_2=2$. Eigenmode transmission is
considered for the PR with $d_1=d_2=2$, and the PR SNR is set as
$20$dB. The channels $\mv{F}$, $\mv{G}_1$, $\mv{G}_2$, and $\mv{H}$
are randomly generated from the standard Rayleigh fading
distribution, and are then fixed in all the examples. The parameters
$\tau$ and $T$ are normalized by the symbol period $T_s$. After the
normalization, $T$ is set as $1,000$ and the lowest value of $\tau$
is set as $10$. The CR transmit rate is measured in nats/complex
dimension (dim.). The peak transmit-power constraint for the CR is
assumed.

\begin{figure}[t]
        \centering
        \includegraphics*[width=9cm]{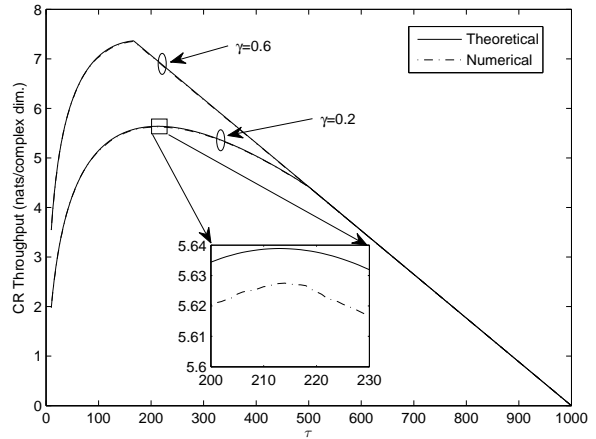}
        \caption{CR throughput versus CR learning time.}
        \label{fig:example1}
\end{figure}

We first fix $P_{\rm CR}$ at CR-Tx as $100$ and show the variations
of the CR throughput as a function of $\tau$. Both theoretical
results obtained in Section \ref{subsec:optimization problem} where
${\hat{\mv U}}$ is not considered as a function of $\tau$ and is
replaced by the true value ${\mv U}$, and numerical results where
$\hat{\mv{U}}$ is obtained via the estimator given in Section
\ref{subsec:channel estimation} with known noise power $\rho_0$, are
shown in Fig. \ref{fig:example1}. The values of $\gamma$ are taken
as $0.2$ and $0.6$, respectively. From Fig. \ref{fig:example1}, the
first observation is that the numerical and theoretical results
almost merge with each other, which supports our previous assumption
of ignoring $\hat{\mv{U}}$ to be a function of $\tau$ during the
optimization process. We also observe that the CR throughput for
$\gamma=0.2$ and that for $\gamma=0.6$ start to merge when $\tau$ is
sufficiently large due to the fact that $g_1(\tau)$ defined in
(\ref{eq:g1}) does not change with $\gamma$. However, the maximum CR
throughput is observed to increase with $\gamma$ because when the
PRs can tolerate more leakage interference powers, the optimal
learning time is reduced and the CR transmit power becomes less
restricted, which leads to an increased CR throughput.

\begin{figure}[t]
        \centering
        \includegraphics*[width=9cm]{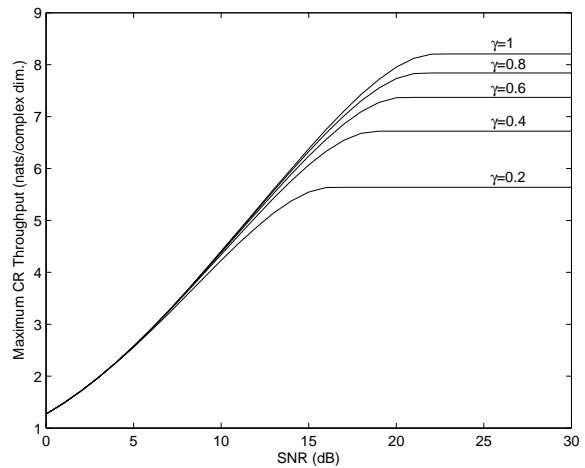}
        \caption{Maximum CR throughput versus CR SNR.}
        \label{fig:example2}
\end{figure}

We then display the maximum CR throughput versus $P_{\rm CR}$, or
equivalently, the CR SNR, in Fig. \ref{fig:example2} for different
values of $\gamma$. Only the theoretical results are shown here. The
first observation is that there exist thresholds on CR SNR, beyond
which the maximum CR throughput cannot be improved for a given
$\gamma$. This is because that when $P_{\rm CR}$ is too large, the
dominant constraint for CR throughput maximization becomes the
interference-power constraint instead of transmit-power constraint.
When this occurs, the intersection point $P_{\rm CR}/\gamma$ in Fig.
\ref{fig:optimal_tau} moves towards $T$. Thus, the optimal value of
$\tau$ and the corresponding maximum CR throughput are determined
from $g_{2}(\tau)$ in (\ref{eq:g2}), which is not related to $P_{\rm
CR}$. Meanwhile, when $\gamma$ increases, the maximum CR throughput
also increases, similarly like the case in Fig. \ref{fig:example1}.

\begin{figure}[t]
        \centering
        \includegraphics*[width=9cm]{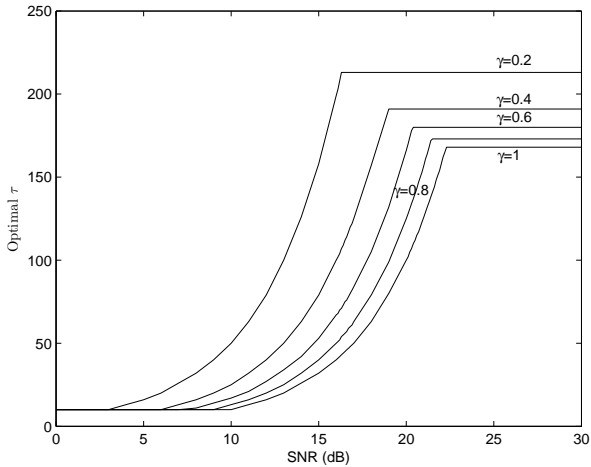}
        \caption{Optimal CR learning time versus CR SNR.}
        \label{fig:example3}
\end{figure}

At last, we show the change of the optimal $\tau$ with respect to
$P_{\rm CR}$ or the CR SNR in Fig. \ref{fig:example3}, where only
the theoretical results are shown. From Fig. \ref{fig:optimal_tau},
we know that when $P_{\rm CR}$ decreases, the intersection point
moves towards zero. Thus, the curves of the optimal learning time
for different $\gamma$'s all merge to the presumed minimum value for
$\tau$, $\tau=10$, at low-SNR region. On the other side, the optimal
values of $\tau$ stop increasing at high-SNR region for a given
$\gamma$, similarly as explained for Fig. \ref{fig:example2}.
Moreover, the optimal $\tau$ is observed to increase with the
decreasing of $\gamma$.

\section{Concluding Remarks}\label{sec:conclusion}

Cognitive beamforming (CB) is a promising technology to enable
high-rate CR transmissions and yet to provide effective interference
avoidance at the coexisting PR terminals. The main challenge for
realizing CB in practical systems is how to obtain the channel
knowledge from CR transmitter to PR terminals. In this paper, we
propose a new solution to this problem utilizing the idea of
effective interference channel (EIC), which can be efficiently
learned at CR transmitter via blind/semiblind estimation over the
received PR signals. Based on the EIC, we design a practical CB
scheme to minimize the effect of the resulted interference on the PR
transmissions. Furthermore, we show that with finite sample size for
channel learning, there exists an optimal learning time to maximize
the CR link throughput.

The developed results in this paper can be readily extended to the
case with multiple PR links. This is so because the proposed CB
scheme is based on the EIC that measures the space spanned by all
the coexisting PR signals as a whole, and thus it works regardless
of these PR signals coming from a single PR link or multiple PR
links.

\appendices

\section{Proof of Lemma \ref{lemma:upperbounds}}\label{appendix:proof upperbounds}

Define ${\mv S}=[{\mv s}(1),\ldots,{\mv s}(N)]$ and ${\mv Y}_s={\mv
{\cal A}}{\mv S}$ where ${\mv s}(n)$'s and ${\mv{\cal A}}$ are given
in Section \ref{sec:effective channel}. From \cite[Appendix I]{Gao},
we know that the first order perturbation\footnote{Note that the
first order approximation is more valid at high-SNR region.} to
${\mv U}$ due to the finite number of samples $N$ and the additive
noise ${\mv Z}\triangleq[{\mv z}(1),\ldots,{\mv z}(N)]$ can be
approximated by
\begin{eqnarray}\label{eq:detlaU}
\Delta {\mv U} \triangleq {\hat {\mv U}}-{\mv U} \approx -({\mv
Y}_s^H)^\dag {\mv Z}^H{\mv U}.
\end{eqnarray}

Since the discussions on $\bar{I}_{1}$ and $\bar{I}_{2}$ are
similar, in the following we restrict our study on $\bar{I}_{1}$.
From the conditions given in Proposition \ref{proposition:1}, we
know that there exists a constant matrix ${\mv W}_1\in{\mathbb
C}^{d_2\times d_1}$, such that ${\mv B}_1{\mv G}_1={\mv W}_1{\mv
A}^H_1{\mv G}_1$. The average interference power, $I_1$ defined in
(\ref{eq:Ij}), is then re-expressed as
\begin{align}
I_1\overset{(a)}{=}& {\mathbb  E}[\mathtt{Tr}({\mv B}_1{\mv
G}_1{\hat{\mv U}}{\mv C}_{CR}{\hat {\mv U}}^H{\mv G}_1^H{\mv B}_1^H)] \nonumber \\
\overset{(b)}{=}& {\mathbb E}[\mathtt{Tr}({\mv B}_1{\mv G}_1\Delta
{\mv U}{\mv C}_{CR}\Delta {\mv U}^H{\mv G}_1^H{\mv B}_1^H)] \nonumber \\
\overset{(c)}{=}& {\mathbb E}[\mathtt{Tr}({\mv B}_1{\mv G}_1({\mv
Y}_s^H)^\dag
{\mv Z}^H{\mv U}{\mv C}_{CR}{\mv U}^H{\mv Z}{\mv Y}_s^\dag {\mv G}_1^H{\mv B}_1^H)] \nonumber \\
\overset{(d)}{=}&\rho_0\mathtt{Tr}({\mv
C}_{CR})\mathbb{E}[\mathtt{Tr}({\mv B}_1{\mv G}_1({\mv Y}_s^H)^\dag
{\mv Y}_s^\dag {\mv G}^H_1{\mv B}_1^H))] \nonumber \\
\overset{(e)}{=}&\rho_0\mathtt{Tr}({\mv
C}_{CR})\mathbb{E}\bigg[\mathtt{Tr}({\mv W}_1{\mv A}^H_1{\mv
G}_1({\mv {\cal A}}^H)^\dag \times \nonumber \\ & ({\mv S}{\mv
S}^H)^{-1} {\mv {\cal A}}^\dag
{\mv G}^H_1{\mv A}_1{\mv W}_1^H)\bigg]\nonumber \\
\overset{(f)}{\approx}& \rho_0\mathtt{Tr}({\mv
C}_{CR})\mathtt{Tr}\left({\mv W}_1[{\mv I}, {\mv 0}]\begin{bmatrix}
\frac{1}{|{\cal N}_1|}{\mv I}&{\mv 0}\\{\mv 0}&\frac{1}{|{\cal
N}_2|}{\mv I}
\end{bmatrix}\begin{bmatrix}{\mv I}\\{\mv 0}\end{bmatrix}{\mv W}_1^H\right) \nonumber \\
=&\frac{\rho_0}{\alpha_1N}\mathtt{Tr}({\mv C}_{CR})\mathtt{Tr}({\mv
W}_1{\mv W}_1^H) \label{eq:I bar final}
\end{align}
where $(a)$ is via substituting (\ref{eq:s CR}) into (\ref{eq:Ij})
and using the independence of ${\hat{\mv U}}$ and $\mv{t}_{\rm
CR}(n)$; $(b)$ is due to ${\mv B}_1{\mv G}_1{\mv U}={\mv 0}$; $(c)$
is due to (\ref{eq:detlaU}); $(d)$ is due to independence of
$\mv{Y}_s$ and $\mv{Z}$ and ${\mathbb E}[{\mv Z}^H{\mv X}{\mv
Z}]=\rho_0\mathtt{Tr}({\mv X}){\mv I}$ for any constant matrix
$\mv{X}$; $(e)$ is due to the definitions of ${\mv W}_1$ and ${\mv
Y}_s$; and $(f)$ is approximately true since $N$ is usually a large
number.

From \cite{Horn}, we have
\begin{align}
& \mathtt{Tr}({\mv W}_1{\mv A}^H_1{\mv G}_1{\mv G}_1^H{\mv A}_1{\mv
W}_1^H)  \nonumber \\ \geq &  \lambda_{\min}({\mv A}^H_1{\mv
G}_1{\mv G}_1^H{\mv A}_1)\mathtt{Tr}({\mv W}_1{\mv W}_1^H)
\label{eq:lambda min}
\end{align}
\begin{align}
\mathtt{Tr}({\mv B}_1{\mv G}_1{\mv G}_1^H{\mv B}^H_1)&\leq
\lambda_{\max}({\mv G}_1{\mv G}_1^H)\mathtt{Tr}({\mv B}_1{\mv
B}_1^H). \label{eq:lambda max}
\end{align}
By noting ${\mv B}_1{\mv G}_1={\mv W}_1{\mv A}^H_1{\mv G}_1$, from
(\ref{eq:lambda min}) and (\ref{eq:lambda max}) it follows that
\begin{equation}\label{eq:inequality}
\mathtt{Tr}({\mv W}_1{\mv W}_1^H)\leq \frac{\lambda_{\max}({\mv
G}_1{\mv G}_1^H)\mathtt{Tr}({\mv B}_1{\mv
B}_1^H)}{\lambda_{\min}({\mv A}^H_1{\mv G}_1{\mv G}_1^H{\mv A}_1)}.
\end{equation}
Using (\ref{eq:Ij bar}), (\ref{eq:I bar final}), and
(\ref{eq:inequality}), the upper bound on $\bar{I}_{1}$ given in
(\ref{eq:UB Ij bar}) is obtained.

\section{Proof of Lemma \ref{lemma:concave}}\label{appendix:proof concavity}

First, it is easily known that $f(z)$ is an increasing function of
$z$. Next, we prove the continuity, differentiability, and concavity
of $f(z)$, respectively.

\subsection{Continuity} From (\ref{eq:f z}), it is known that in
each section $[q_{k-1},q_{k}]$, $f(z)$ is obviously continuous. For
boundary points of each section, we have
\begin{align}
\lim_{z\rightarrow q_k^{-}}f(z)&=
\sum_{i=1}^k\log\left(\frac{\sigma^2_{h,i}}{\sigma^2_{h,k+1}}\right)=\lim_{z\rightarrow
q_k^{+}}f(z),
\end{align}
$k=1,\ldots,M_t-d_{\rm eff}-1$. Thus, $f(z)$ is continuous at all
the points.

\subsection{Differentiability} From (\ref{eq:f z}), it is known
that in each section $[q_{k-1},q_{k}]$, $f(z)$ is differentiable.
For boundary points of each section, it can be verified that
\begin{align}
\lim_{z\rightarrow q_k^{-}}{\dot
f}(z)&=\frac{\sigma_{h,k+1}^2}{\rho_1}=\lim_{z\rightarrow
q_k^{+}}{\dot f}(z),
\end{align}
$k=1,\ldots,M_t-d_{\rm eff}-1$. Therefore, $f(z)$ is differentiable
at all the points.

\subsection{Concavity} For a given $z$, $f(z)$ is obtained by
solving the optimization problem in (\ref{eq:optimize f z}), which
can be easily verified to be a convex optimization problem
\cite{Boydbook}. Thus, the duality gap for this optimization problem
is zero and $f(z)$ can be equivalently obtained as the optimal value
of the following min-max optimization problem:
\begin{align}
f(z)&=\min_{\mu\geq0}\max_{x_i\geq0}
\sum_{i}\log(1+\frac{\sigma^2_{h,i}x_i}{\rho_1})-\mu(\sum_{i}x_i-z)\\
&=\min_{\mu\geq0}
\sum_{i}(\log(\frac{\sigma^2_{h,i}}{\rho_1\mu}))^+-\sum_{i}(1-\frac{\rho_1\mu}{\sigma^2_{h,i}})^+ +\mu z\\
&=\sum_{i}(\log(\frac{\sigma^2_{h,i}}{\rho_1\mu^{(z)}}))^+
-\sum_{i}(1-\frac{\rho_1\mu^{(z)}}{\sigma^2_{h,i}})^+ +\mu^{(z)}z
\end{align}
where the summations are taken over $i=1,\ldots,M_t-d_{\rm eff}$,
and $\mu^{(z)}\geq 0$ is the optimal dual variable for a given $z$.
In fact, it can be shown that $1/\mu^{(z)}$ is just the water level
given in (\ref{eq:water level}) corresponding to the total power
$z$.

Denote $\omega$ as any constant in $[0,1]$. Let $\mu^{(z_1)}$,
$\mu^{(z_2)}$, and $\mu^{(z_3)}$ be the optimal $\mu$ for $f(z_1)$,
$f(z_2)$, and $f(z_3), z_3=\omega z_1+(1-\omega)z_2$, respectively.
For $j=1,2$, we have
\begin{align}
f(z_j)&=\sum_{i}(\log(\frac{\sigma^2_{h,i}}{\rho_1\mu^{(z_j)}}))^+
-\sum_{i}(1-\frac{\rho_1\mu^{(z_j)}}{\sigma^2_{h,i}})^+ +\mu^{(z_j)}z_j\\
&\leq\sum_{i}(\log(\frac{\sigma^2_{h,i}}{\rho_1\mu^{(z_3)}}))^+
-\sum_{i}(1-\frac{\rho_1\mu^{(z_3)}}{\sigma^2_{h,i}})^+
+\mu^{(z_3)}z_j
\end{align}
where the inequality is due to the fact that $\mu^{(z_3)}$ is not
the optimal dual solution for $j=1,2$. Therefore,
\begin{align}
& \omega f(z_1)+(1-\omega)f(z_2)\nonumber \\ \leq&
\sum_{i}(\log(\frac{\sigma^2_{h,i}}{\rho_1\mu^{(z_3)}}))^+
-\sum_{i}(1-\frac{\rho_1\mu^{(z_3)}}{\sigma^2_{h,i}})^+ +\mu^{(z_3)}z_3\\
=&f(z_3) \\=&f(\omega z_1+(1-\omega)z_2).
\end{align}
Thus, $f(z)$ is a concave function \cite{Boydbook}.

\end{document}